\documentclass[11pt]{article}

\usepackage[a4paper, hmargin=0.8in, top=1in, bottom=1.1in, footskip=0.6in]{geometry}

\usepackage{amsmath}
\usepackage{amssymb}
\usepackage{amsfonts}
\usepackage{latexsym}
\usepackage{amsthm}
\usepackage{amsxtra}
\usepackage{amscd}
\usepackage{mathrsfs}
\usepackage{mathscinet}
\usepackage{bm}
\usepackage{graphicx}

\usepackage{color}
\usepackage{cite}
\usepackage{hyperref}

\theoremstyle{plain} 
\newtheorem{theorem}{Theorem}
\newtheorem*{theorem*}{Theorem}
\newtheorem{lemma}[theorem]{Lemma}
\newtheorem*{lemma*}{Lemma}

\newtheorem*{corollary*}{Corollary}
\newtheorem{fact}[theorem]{Fact}
\newtheorem{proposition}[theorem]{Proposition}
\newtheorem*{proposition*}{Proposition}

\newtheorem*{definition*}{Definition}

\newtheorem*{conjecture*}{Conjecture}

\newtheorem{assumption}[theorem]{Assumption}

\newtheorem*{example*}{Example}

\newtheorem*{remark*}{Remark}

\newcommand{\wh}{\widehat}

\definecolor{darkred}{rgb}{0.9,0,0.3}
\definecolor{darkblue}{rgb}{0,0.3,0.9}

\usepackage{ifthen}
\def\comment#1{\ifthenelse{\isodd{\value{page}}}{\marginpar{\raggedright\scriptsize{\textcolor{darkred}{#1}}}}{\marginpar{\raggedleft\scriptsize{\textcolor{darkred}{#1}}}}}  

\newcommand{\cB}{\ensuremath{\mathcal B}}

\newcommand{\cE}{\ensuremath{\mathcal E}}

\newcommand{\cG}{\ensuremath{\mathcal G}}
\newcommand{\cH}{\ensuremath{\mathcal H}}
\newcommand{\cI}{\ensuremath{\mathcal I}}
\newcommand{\cJ}{\ensuremath{\mathcal J}}

\newcommand{\cZ}{\ensuremath{\mathcal Z}}

\newcommand{\bbE}{{\ensuremath{\mathbb E}} }

\newcommand{\bbV}{{\ensuremath{\mathbb V}} }

\newcommand{\Id}{\text{Id}}

\newcommand{\R}{\mathbb{R}}
\newcommand{\C}{\mathbb{C}}

\newcommand{\N}{\mathbb{N}}
\newcommand{\Z}{\mathbb{Z}}

\newcommand{\ee}{\mathrm{e}}
\newcommand{\ii}{\mathrm{i}}

\newcommand{\id}{\mspace{2mu}\mathrm{i}\mspace{-0.6mu}\mathrm{d}} 

\renewcommand{\geq}{\geqslant}
\renewcommand{\epsilon}{\varepsilon}

\newcommand{\scalara}[2]{\left\langle{#1} \,\mspace{2mu},\, {#2}\right\rangle}

\DeclareMathOperator{\diag}{diag}
\DeclareMathOperator{\tr}{Tr}
\DeclareMathOperator{\var}{var}

\DeclareMathOperator{\supp}{supp}
\DeclareMathOperator{\real}{Re}

\DeclareMathOperator{\diam}{diam}

\DeclareMathOperator{\size}{size}

\newcommand{\EE}{\mathrm{E}}

\newcommand{\ve}{\vee}

\newcommand{\ceins}{c_{4}}
\newcommand{\czwei}{c_{3}}
\newcommand{\cdrei}{c_{6}}
\newcommand{\cfuenf}{c_{5}}
\newcommand{\csieben}{c_{8}}
\newcommand{\czehn}{c_{11}}
\newcommand{\celf}{{c_{2}}}
\newcommand{\cdreizehn}{c_{12}}
\newcommand{\cvierzehn}{c_{13}}
\newcommand{\cfuenfzehn}{c_{14}}
\newcommand{\csechzehn}{c_{15}}
\newcommand{\csiebzehn}{c_{18}}
\newcommand{\cachtzehn}{c_{19}}
\newcommand{\cneunzehn}{c_{16}}
\newcommand{\czwanzig}{c_{17}}
\newcommand{\ceinundzwanzig}{{c_{1}}}
\newcommand{\czweiundzwanzig}{c_{9}}
\newcommand{\cdreiundzwanzig}{c_{10}}
\newcommand{\cvierundzwanzig}{c_{7}}

\begin{document}
\thispagestyle{empty}
\begin{center}
{\bf\Large Dislocation lines in three-dimensional solids at low temperature}
\\[3mm]
Roland Bauerschmidt\footnote{Statistical Laboratory, DPMMS, University of Cambridge, 
Wilberforce Road, Cambridge, CB3 0WB, United Kingdom. 
E-mail: rb812@cam.ac.uk
}
\hspace{1mm} 
Diana Conache\footnote[2]{%
Technische Universit{\"{a}}t M{\"{u}}nchen, 
Zentrum Mathematik, Bereich M5,
D-85747 Garching bei M{\"{u}}nchen,
Germany.
E-mail: diana.conache@tum.de, srolles@ma.tum.de}
\hspace{1mm} 
Markus Heydenreich\footnotemark[3]
\hspace{1mm} 
Franz Merkl\footnote[3]{Mathematical Institute, Ludwig-Maximilians-Universit\"at M\"unchen,
Theresienstr.\ 39,
D-80333 Munich,
Germany.
E-mail: m.heydenreich@lmu.de, merkl@math.lmu.de
}
\hspace{1mm} 
Silke W.W.\ Rolles\footnotemark[2]
\\[3mm]
{\small July 31, 2019}\\[3mm]
\end{center}

\begin{abstract}
We propose a model for three-dimensional solids on a mesoscopic scale with a statistical mechanical description of dislocation lines in thermal equilibrium. 
The model has a linearized rotational symmetry, which is broken by boundary conditions. We show that this symmetry is spontaneously broken in the thermodynamic limit at small positive temperatures. \\[2mm]
\end{abstract}

  
\section{Introduction}

\subsection{Motivation and background}

The perhaps most fundamental mathematical problem of solid state physics is that of crystallization,
which in a classical version could be formulated as follows.
Let $v: \R_{>0}\to \R$ be a two-body potential of Lennard--Jones type, and consider
the $N$-particle energy
\begin{equation}
  H_N(x_1,\dots, x_N) = \sum_{i \neq j} v(|x_i-x_j|), \quad (x_i \in \R^3).
\end{equation}
The Grand Canonical Gibbs measure at inverse temperature $\beta>0$ and fugacity $z>0$ in finite volume $\Omega \subset \R^3$ is the point process on $\Omega$ 
given by
\begin{equation}
  P_{\Omega,\beta,z}(B) = \frac{1}{Z_{\Omega,\beta,z}}\sum_{N\in\N_0}
\int_{\Omega^N \cap B} \frac{z^N}{N!} \ee^{-\beta H_{N}(x)} \, dx 
\end{equation}
for measurable $B \subseteq \bigcup_{N\in\N_0}\Omega^N$
modulo permutations of particles with the appropriate normalizing constant
$Z_{\Omega,\beta,z}$.
The crystallization problem is to prove that at low temperature and high density,
i.e., large but finite inverse temperature $\beta$ and fugacity $z$,
there exist corresponding infinite volume Gibbs measures that are non-trivially periodic with the symmetry of a crystal lattice.
This problem remains far out of reach.
Unfortunately,
even the zero temperature case, i.e., studying the limiting minimizers of the finite volume energy, is open and appears to be very difficult.
For the zero temperature case in two dimensions, see Theil \cite{MR2200888} and references therein;
for important progress on the problem in three dimensions, see Flatley and Theil \cite{MR3360741} and references therein. 
Detailed understanding of the zero temperature case is a prerequisite for understanding the low temperature regime,
but in addition, any description at finite temperature must
explain spontaneous breaking of the rotational symmetry and take into account the possibility of crystal dislocations.
This significantly complicates the problem because proving
spontaneous breaking of continuous symmetries is already notoriously difficult in models where
the ground states are obvious, such as in the $O(3)$ spin model,
for which the only robust method is the very difficult work of Balaban, 
see \cite{MR1669693} and references therein.

With a realistic microscopic model out of reach,
we  start from a mesoscopic rather than microscopic perspective to understand the effect of the dislocations.
We expect the latter to be a fundamental aspect also of the original problem.
Our model does not account for the full rotational symmetry, but is only invariant under linearized rotations.
More precisely, our model for deformed solids in three dimensions
consists of a gas of closed vector-valued defect lines which describe crystal dislocations on a mesoscopic scale.
For this model, we show that the breaking of the linearized rotational symmetry persists in the thermodynamic limit.

Our model is strongly motivated by the one introduced and studied by Kosterlitz and Thouless \cite{0022-3719-6-7-010},
and refined by Nelson and Halperin \cite{PhysRevB.19.2457}, and Young \cite{PhysRevB.19.1855}.
The Kosterlitz--Thouless model has an energy that consists of an elastic contribution and a contribution due to crystal dislocations. These two contributions 
are assumed independent.
This KTHNY theory explains crystallization and the melting transition in two dimensions,
as a transition mediated by vector-valued dislocations effectively interacting through a Coulomb interaction.
For a textbook treatment of this phenomenology, see Chaikin and 
Lubensky \cite{chaikin_lubensky_1995}.
The Kosterlitz--Thouless model for two-dimensional melting is closely related to the two-dimensional rotator model,
studied by Kosterlitz and Thouless in the same paper as the melting problem,
following previous insight by Berezinski\u{\i} \cite{MR0314399}.
For their model of a two-dimensional solid, the assumption that the energy consists of elastic and dislocation
contributions, which can be assumed to be essentially independent, is not derived from a realistic microscopic model.
On the other hand, the rotator model admits an exact description in terms of spin waves (corresponding to the elastic energy)
and vortices described by a scalar Coulomb gas.
In this description,
the spin wave and vortex contributions are not far from independent,
and, in fact, in the Villain version of the rotator model \cite{villain:jpa-00208289}, they become exactly independent.
Based on a formal renormalization group analysis,
Kosterlitz and Thouless proposed a novel phase transition mediated by unbinding of the topological defects,
the Berezinski\u{\i}--Kosterlitz--Thouless transition.
In the two-dimensional rotator model the existence of this transition was proved by Fr\"ohlich and Spencer \cite{MR634447}.
For recent results on the two-dimensional Coulomb gas, see Falco \cite{MR2917175}.
In higher dimensions, the description of the rotator model in terms of spin waves and vortex defects remains valid,
except that the vortex defects, which are point defects in two dimensions, now become closed vortex lines \cite{MR649811}
as in our solid model.
Using this description and the methods they had introduced for the two-dimensional case,
Fr\"ohlich and Spencer \cite{MR634447,MR649811}
proved long-range order for the rotator model at low temperature in dimensions $d\geq 3$,
without relying on reflection positivity.
The latter is a very special feature used in \cite{MR0421531} to establish
long-range order for the $O(n)$ model exactly with nearest neighbor interaction on $\Z^d$, $d\geq 3$.
In general, proving spontaneous symmetry breaking of continuous symmetries remains a difficult problem.
However, aside from the most general approach of Balaban and reflection positivity,
for abelian spin models, several other techniques exist \cite{MR649811,MR836009}.

As discussed above, our model, defined precisely in \eqref{e:Pbeta-def},
is closely related to the Kosterlitz--Thouless model, see for example 
\cite[(9.5.1)]{chaikin_lubensky_1995}.
Our analysis is based on the Fr\"ohlich--Spencer approach for the rotator model \cite{MR634447,MR649811}.

In a parallel study of Giuliani and Theil \cite{GT} following 
Ariza and Ortiz \cite{MR2186424}, a model very similar 
to ours is examined, but with a microscopic interpretation, describing locations 
of individual atoms. In particular, it also has a linearized rotational symmetry. 

In \cite{MR2535010,MR3296527} (see also \cite{MR3357973}), some of us studied other simplified models for crystallization.
These models have full rotational symmetry,
but do not permit dislocations. In \cite{MR2535010} defects were excluded, while in the model in \cite{MR3296527} 
isolated missing single atoms were allowed.

\subsection{A linear model for dislocation lines on a mesoscopic scale}

\subsubsection{Linearized elastic deformation energy}
\label{sec:elastic}

An elastically deformed solid in continuum approximation can be described 
by a deformation map $f:\R^3\to\R^3$ with the interpretation 
that for any point $x$ in the undeformed solid $f(x)$ is the location of 
$x$ after deformation.
The Jacobi matrix $\nabla f:\R^3\to\R^{3\times 3}$ describes the deformation map 
locally in linear approximation. Only orientation preserving maps, $\det\nabla f>0$,
make sense physically.
The elastic deformation energy $E_{\mathrm{el}}(f)$ 
is modeled to be an integral over a smooth elastic energy density 
$\rho_{\mathrm{el}}:\R^{3\times 3}\to\R$ (respectively $\tilde\rho_{\mathrm{el}}:\R^{3\times 3}\to\R$):
\begin{align}
  E_{\mathrm{el}}(f)
  =\int_{\R^3}\rho_{\mathrm{el}}(\nabla f(x)) \, dx
  =\int_{\R^3}\tilde\rho_{\mathrm{el}}((\nabla f)^t(x)\, \nabla f(x)) \, dx,
\end{align}
where the second representation holds under the assumption of \emph{rotation invariance}; see Appendix~\ref{sec:elastic2}.
From now on, we consider only small perturbations $f=\id+\epsilon u:\R^3\to\R^3$ 
of the identity map as deformation maps. The parameter $\epsilon$ corresponds 
to the ratio between the microscopic and the mesoscopic scale. 
We Taylor-expand $\tilde\rho_{\mathrm{el}}$
around the identity matrix $\Id$ using that $\tilde\rho_{\mathrm{el}}$ is smooth near $\Id$, obtaining
\begin{equation} \label{e:F-first}
\tilde\rho_{\mathrm{el}}(\Id+\epsilon A)
=\epsilon^2F(A)+O(\epsilon^3),\quad (\epsilon\to 0,\; A=A^t\in\R^{3\times 3}), 
\end{equation}
with a positive definite quadratic form $F$ on symmetric matrices.
Under the assumption of \emph{isotropy} (see Appendix~\ref{app:elasticity}), 
writing $|{\cdot}|$ for the Euclidean norm, the general form for $F$ is
\begin{equation}
\label{eq:def-F}
  F(U)=\frac{\lambda}{2}(\tr U)^2+\mu|U|^2 
\qquad \text{with }
\mu>0\text{ and }\mu+3\lambda/2>0.  
\end{equation}
In elasticity theory, the constants $\lambda$ and $\mu$ are the so-called Lam\'e coefficients.
Even for cubic monocrystals, the isotropy assumption is restrictive for realistic models.
While it is not important for our ana\-ly\-sis, we nonetheless assume isotropy to keep the notation somewhat simpler.
We refer to \cite[Chapters 6.4.2 and 6.4.3]{chaikin_lubensky_1995} for a discussion on the number of elastic constants necessary in order to describe various crystal systems.
Summarizing, we have the following model for the linearized elastic
deformation energy:
\begin{align}
  \label{eq:def_E_el}
 E_{\mathrm{el}}(\id+\epsilon u)&=\epsilon^2 H_{\mathrm{el}}(\nabla u)+O(\epsilon^3)
\quad\text{with}
  \\
H_{\mathrm{el}}(w)&=\int_{\R^3}F(w(x)+w(x)^t) \, dx,
\label{eq:def_H_el}
\end{align}
for measurable $w:\R^3\to\R^{3\times 3}$ and 
for $F$ as in \eqref{eq:def-F}. 

\subsubsection{Burgers vector densities}

The following model is intended to describe dislocation lines
on a mesoscopic scale
as they appear in solids at positive temperature. 
We describe the solid by a smooth map $w:\R^3\to\R^{3\times 3}$
replacing the map $\nabla u:\R^3\to\R^{3\times 3}$ from
Section~\ref{sec:elastic}.
If dislocation lines are absent, the model described now
boils down to the setup of Section~\ref{sec:elastic} with $w=\nabla u$ being
a gradient field. The field 
\begin{align}
b:\R^3\to\R^{3\times 3\times 3},\quad
b_{ijk}=(d_1w)_{ijk}:=\partial_iw_{jk}-\partial_jw_{ik}
\end{align}
is intended to describe the Burgers vector density.
It vanishes if and only if $w=\nabla u$ is a gradient field.
One can interpret $b_{ijk}$ as the $k$-th component of the resulting 
vector per area if one goes through the image in the deformed solid of a 
rectangle which is parallel to the $i$-th and $j$-th coordinate axis. The 
antisymmetry $b_{ijk}=-b_{jik}$ can be interpreted as the change of sign if 
the orientation of the rectangle is changed. 
Any smooth field
$b:\R^3\to\R^{3\times 3\times 3}$
which is antisymmetric in its first two indices 
is of the form
$b=d_1w$ with some $w:\R^3\to\R^{3\times 3}$ if and only if
\begin{align}
\label{eq:db=0}
d_2b=0,
\end{align}
where
\begin{align}
d_2b:\R^3\to\R^{3\times 3\times 3\times 3},\quad
(d_2b)_{lijk}=\partial_lb_{ijk}+\partial_ib_{jlk}+\partial_jb_{lik}
\end{align}
denotes the exterior derivative with respect to the first two indices.
Being antisymmetric in its first two indices,
it is convenient to write the Burgers vector density $b$
in the form
\begin{align}
\label{eq:representation-b-in-terms-of-tilde-b}
b_{ijk}=\sum_{l=1}^3 \epsilon_{ijl}\tilde{b}_{lk},
\end{align}
where $\tilde{b}:\R^3\to\R^{3\times 3}$
and $\epsilon_{ijk}=\det(e_i,e_j,e_k)$ with the standard unit vectors 
$e_i\in\R^3$, $i\in[3]:=\{1,2,3\}$. 
The integrability condition
\eqref{eq:db=0} can be written in the form
\begin{align}
\label{eq:div-tilde-b-zero}
\operatorname{div} \tilde b=0
\quad\text{ with }\quad
(\operatorname{div} \tilde b)_k:=\sum_{l=1}^3 \partial_l\tilde{b}_{lk}.
\end{align}
In view of this equation, one may visualize $\tilde b$ to 
be a sourceless vector-valued current.

\subsubsection{Model assumptions}
\label{subsubsec:model-assum}
In linear approximation,
the leading order total energy of a deformed solid
described by $w:\R^3\to\R^{3\times 3}$ is modeled to
consist of an ``elastic'' part and a local ``dislocation'' part:
\begin{align}
H(w)&=
H_{\mathrm{el}}(w)
+\cH_{\mathrm{disl}}(d_1w),
\end{align}
where $H_{\mathrm{el}}(w)$ was introduced in \eqref{eq:def_H_el}.
The field $w$ consists of an exact contribution (modeling purely elastic fluctuations) and a coexact contribution representing the elastic part of the energy induced by dislocations.  Both these contributions are contained in $H_{\mathrm{el}}(w)$, while $\cH_{\mathrm{disl}}(d_1w)$ is intended to model only the \emph{local} energy of dislocations. 
The dislocation part $\cH_{\mathrm{disl}}(b)\in[0,\infty]$ is defined for measurable 
$b:\R^3\to\R^{3\times 3\times 3}$ being antisymmetric in its first two indices.

We describe now a formally coarse-grained model for dislocation lines:
Dislocation lines are only allowed in the set $\Lambda$ of undirected edges of 
a \textit{mesoscopic} lattice in $\mathbb R^3$. Let $V_\Lambda$
denote its vertex set. As a lattice, the graph 
$(V_\Lambda,\Lambda)$ is of bounded degree.
To model boundary conditions, we only allow dislocation lines on 
a finite subgraph $G=(V,E)$ of $(V_\Lambda,\Lambda)$, ultimately taking the 
thermodynamic limit $E\uparrow\Lambda$.
We write $E\Subset\Lambda$ if $E$ is a finite subset of $\Lambda$. 
We denote the edge between adjacent vertices $x,y\in V_\Lambda$ by $\{x,y\}$. 
The graph $(V_\Lambda,\Lambda)$ is {\it not} intended to describe the atomic
structure of the solid, as it lives on a mesoscopic scale.
Rather, it is just a tool to introduce a coarse-grained
structure which eventually makes the model discrete.

To every edge $e=\{x,y\}$, we associate a counting direction,
which has no physical meaning but serves only for bookkeeping purposes.
The Burgers vectors on the finite subgraph $G=(V,E)$ are encoded by a
family $I=(I_e)_{e\in E}\in(\R^3)^E$ of vector-valued currents
flowing through the edges in counting direction. A vector $I_e$ means 
the Burgers vector associated to a closed curve surrounding the dislocation 
line segment $[x,y]$ in positive orientation with respect to the counting 
direction. The family of currents $I$ 
should fulfill Kirchhoff's node law
\begin{align}
\label{eq:Kirchhoff}
\sum_{e\in E}s_{ve}I_e=0,\quad(v\in V),
\end{align}
where $s\in\{1,-1,0\}^{V_\Lambda\times \Lambda}$ is the signed incidence matrix
of the graph $(V_\Lambda,\Lambda)$, defined by its entries
\begin{align}
s_{ve}=\begin{cases}
1 &\text{ if $e$ is an incoming edge of $v$,}\\ 
-1&\text{ if $e$ is an outgoing edge from $v$,}\\
0&\text{ otherwise.}
\end{cases}
\end{align}
The distribution of the current in space encoded by $I$ is supported on the union 
of 
the line segments $[x,y]$, with $\{x,y\}\in E$. Thus, it is a rather singular 
object having no density with respect to the Lebesgue measure on $\R^3$. We describe 
it as follows by a matrix-valued measure $J(I):\operatorname{Borel}(\R^3)\to\R^{3\times 3}$ 
on $\R^3$, supported on the union of all edges: For $e=\{x,y\}\in \Lambda$,
let $\lambda_e:\operatorname{Borel}(\R^3)\to\R_{\ge 0}$ denote the 
1-dimensional Lebesgue measure on the line segment $[x,y]$; 
it is normalized by $\lambda_{e}(\R^3)=|x-y|$. Furthermore, let 
$n_e\in\R^3$ denote the unit vector pointing in the counting direction of the edge $e$. 
The matrix-valued measure $J(I)$ is then defined by
\begin{align}
\label{eq:def-J-jk-I}
\operatorname{Borel}(\R^3)\ni B\mapsto
J_{jk}(I)(B)=\sum_{e\in E} (n_e)_j(I_e)_k\lambda_e(B),\quad j,k\in[3].
\end{align}
Thus, the index $k$ encodes a component of the Burgers vector 
and the index $j$ a component of the direction of the dislocation line.

Heuristically, the current distribution $J(I)$ is intended to describe a 
coarse-grained picture of a much more complex microscopic dislocation configuration:
On an elementary cell of the mesoscopic lattice this microscopic configuration 
is replaced by a vector-valued current on a single dislocation line segment, encoding the effective 
Burgers vector. Because the outcome $J$ of this heuristic coarse-graining procedure
is such a singular object supported on line segments, its elastic energy close to 
the dislocation lines would be ill-defined. Hence, the coarse-graining must be 
accompanied by a smoothing operation. 

More precisely, the Burgers vector density $\tilde{b}(I)$ associated to $I$
is modeled by the convolution of $J(I)$
with a form function $\varphi$:
\begin{align}\label{eq:def-b-I}
\tilde{b}(I)&=\varphi\ast J(I).
\end{align}
Here, the form function $\varphi:\R^3\to\R_{\ge 0}$ is chosen to be smooth,
compactly supported, with total mass $\|\varphi\|_1=1$, and $\varphi(0)>0$. 

Altogether, this yields the Burgers vector density
as a function of $I$:
\begin{align}
\label{eq:b-as-fn-of-I}
b_{ijk}(I)=\sum_{l=1}^3 \epsilon_{ijl}\varphi\ast J_{lk}(I)
=\sum_{l=1}^3 \epsilon_{ijl}\varphi\ast 
\sum_{e\in E} (n_e)_l(I_e)_k\lambda_e.
\end{align}  
For a graphical illustration of $I$ and $b(I)$ see Figure 1.

\begin{figure}[h]
\centering\includegraphics[width=0.2\textwidth, angle =-90] {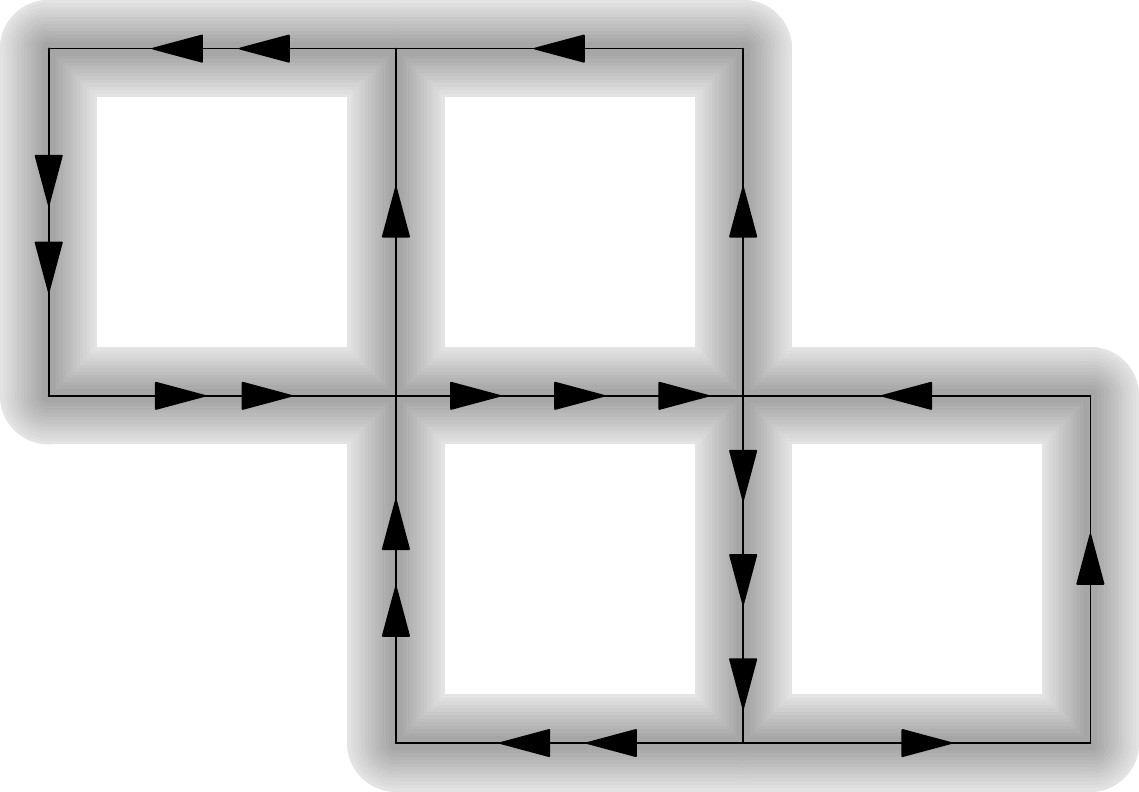}
\caption{Possible components $J_{j k}(I)$ as well as their ``smoothed versions'' $\tilde b_{j k}(I)$ shaded gray.}
\end{figure}

It is shown in Appendix~\ref{sec:appendix-kirchhoff} that the Kirchhoff node law 
\eqref{eq:Kirchhoff} implies that $\tilde b(I)$ is 
sourceless, i.e., equation \eqref{eq:div-tilde-b-zero} holds for it, or equivalently, 
that the integrability condition $d_2b=0$ is valid. 

We now impose an additional discreteness condition on $I$,
which encodes the restriction that Burgers vectors should
take values in a \textit{microscopic} lattice reflecting the atomic 
structure of the solid.
Let $\Gamma\subset\R^3$ be a lattice, interpreted as the
microscopic lattice (scaled to length scale 1).
We set
\begin{align}
\label{eq:def-cal-I}
\cI=\cI(E)=\{I\in\Gamma^E:\;\text{\eqref{eq:Kirchhoff} holds for $I$}\}.
\end{align}
Note that the current $I$ is indexed by the edges in the mesoscopic  
graph $(V,E)$ but 
takes values in the microscopic lattice $\Gamma$. One should not
confuse the mesoscopic graph $(V,E)$ nor the mesoscopic lattice $\Lambda$ 
with the microscopic lattice $\Gamma$;
they have nothing to do with each other.
A motivation for the introduction of two different lattices $\Gamma$ and 
$V_\Lambda$ on two different length scales is described in the discussion of 
the model at the end of this section.

From now on, we abbreviate 
$H_{\mathrm{disl}}(I):=\cH_{\mathrm{disl}}(b(I))$ and $\supp I:=\{e\in E:I_e\neq 0\}$. 
We require the following general assumptions: 
\begin{assumption}
\label{ass:H-disl}
\begin{itemize}
\item \textit{Symmetry:} $H_{\mathrm{disl}}(I)=H_{\mathrm{disl}}(-I)$ for all $I\in\cI$.
\item \textit{Locality:} For $I=I_1+I_2$ with $I_1,I_2\in\cI$ such that no edge in $\supp I_1$ has a common vertex with another edge in $\supp I_2$
we have $H_{\mathrm{disl}}(I)=H_{\mathrm{disl}}(I_1)+H_{\mathrm{disl}}(I_2)$.
Moreover, $H_{\mathrm{disl}}(0)=0$.
\item \textit{Lower bound:} For some constant $c>0$ and all $I\in \cI$, 
\begin{align}
\label{eq:lower-bound-H-disl-in-terms-of-I}
H_{\mathrm{disl}}(I)\ge c\|I\|_1 := c  \sum_{e\in E} |I_e|.
\end{align}
\end{itemize}
\end{assumption}
Condition \eqref{eq:lower-bound-H-disl-in-terms-of-I} reflects the local energetic costs of dislocations in addition to the elastic energy costs reflected by $H_{\rm el}(w)$.

One obtains typical examples for $H_{\mathrm{disl}}(I)$ by requiring a number of assumptions on the form function $\varphi$ and $\cH_{\mathrm{disl}}(b)$: 
First, for all $\{u,v\},\{x,y\}\in\Lambda$ with $\{u,v\}\cap\{x,y\}=\emptyset$, 
we assume
$([u,v]+\supp\varphi)\cap ([x,y]+\supp\varphi)=\emptyset$.
Furthermore, 
\begin{align}
\label{eq:condition-supp-varphi}
\inf_{e=\{x,y\}\in\Lambda}\lambda_e\big(\big\{z\in[x,y]: & (z+\supp\varphi)\cap(w+\supp\varphi)=\emptyset \nonumber\\
& \forall w\in[u,v]\text{ with }\{u,v\}\in\Lambda\setminus\{e\}\big\}\big)> 0.  
\end{align}
Roughly speaking, the last condition means that different
edges $e\in\Lambda$ do not overlap too much after broadening with $\supp\varphi$.
Finally, for all $b$, 
\begin{align}\label{eq:H_disl_lower_bound_b}
	\cH_{\mathrm{disl}}(b)=\cH_{\mathrm{disl}}(-b)&\ge\ceinundzwanzig\int_{\R^3}|b(x)|_1\,dx
	=\ceinundzwanzig\sum_{i,j,k\in[3]}\int_{\R^3}|b_{ijk}(x)|\,dx 
\end{align}
for some constant $\ceinundzwanzig>0$, and for $b=b_1+b_2$ with 
$\supp b_1\cap\supp b_2=\emptyset$,  
it is true that $\cH_{\mathrm{disl}}(b)=\cH_{\mathrm{disl}}(b_1)+\cH_{\mathrm{disl}}(b_2)$.
One particular example is obtained by taking equality in \eqref{eq:H_disl_lower_bound_b}.

\subsection{Model and main result}
\label{sec:hamiltonian-of-model}

One summand in the Hamiltonian of our model is defined by 
\begin{align}\label{eq:Hel-star-def}
H^*_{\rm el}(I)=\inf_{\substack{w\in C^\infty_c(\R^3,\R^{3\times 3}):\\d_1w= b(I)}}
H_{\rm el}(w)
\end{align}
for all $I\in(\R^3)^E$ satisfying \eqref{eq:Kirchhoff}.  
The condition that $w$ is compactly supported reflects the boundary condition, in the
sense that close to infinity the solid must not be moved away from its reference 
location. The symmetry with respect to linearized global rotations is reflected
by the fact $H_{\rm el}(w)=H_{\rm el}(w+w_{\rm{const}})$ for every constant 
antisymmetric matrix $w_{\rm{const}}\in\R^{3\times 3}$. Only the boundary condition, 
i.e., only the restriction that $w$ should have compact support, 
breaks this global symmetry. This paper is about the question whether 
this symmetry breaking persists in the thermodynamic limit $E\uparrow\Lambda$. 

Because $H_{\rm el}$ is positive semidefinite, $H^*_{\rm el}$ is positive semidefinite as well, cf.\ \eqref{eq:claim-minimizer} below. 
This gives us the following linearized model for the dislocation lines at inverse temperature $\beta<\infty$:
\begin{align}
\label{e:Zbeta-def}  
Z_{\beta,E}:=&\sum_{I\in \mathcal{I}(E)}
\ee^{-\beta(H^*_{\mathrm{el}}(I)+H_{\mathrm{disl}}(I))},\\
\label{e:Pbeta-def}  
P_{\beta,E}:=&\frac{1}{Z_{\beta,E}}\sum_{I\in \mathcal{I}(E)}
\ee^{-\beta(H^*_{\mathrm{el}}(I)+H_{\mathrm{disl}}(I))}\delta_I,
\end{align}
where $\delta_I$ denotes the Dirac measure in $I\in\mathcal{I }$ and 
we use the convention $\ee^{-\infty}=0$ throughout the paper.
Whenever the $E$-dependence is kept fixed, we use the abbreviations 
$Z_\beta=Z_{\beta,E}$ and $P_\beta=P_{\beta,E}$. 

The following preliminary result shows that any sequence of smooth configurations satisfying the boundary conditions (i.e., being compactly supported) with prescribed Burgers vectors has a limit $w^*$ in $L^2$ provided that the energy is approaching the infimum of all energies within the class. We show later that this limit is a unique minimizer 
of $H_{\rm el}$ in a suitable Sobolev space. 
An explicit description of $w^*$
is provided in Lemma~\ref{le:minimizer-elastic-energy} below.

\begin{proposition}[Compactly supported approximations of the minimizer]
\label{pro:convergence-to-w-star}
For any $I\in\cI$, there is a bounded smooth function 
$w^*(\cdot,I)\in L^2(\R^3,\R^{3\times 3})$
such that for any sequence $(w^n)_{n\in\N}$ in 
$C^\infty_c(\R^3,\R^{3\times 3})$ with $d_1w^n=b(I)$ for all $n\in\N$ and 
$\displaystyle\lim_{n\to\infty}H_{\rm el}(w^n)=H^*_{\rm el}(I)$
we have $\displaystyle\lim_{n\to\infty}\|w^n-w^*(\cdot,I)\|_2=0$.
\end{proposition}
In the whole paper, constants are denoted by 
$c_1,c_2,$ etc. They may depend on the fixed model ingredients: 
the microscopic lattice $\Gamma$, the mesoscopic lattice $\Lambda$, the 
constant $c$ from
formula~\eqref{eq:lower-bound-H-disl-in-terms-of-I}, and the form function
$\varphi$. All constants keep their meaning throughout the paper.
Similarly, the expression ``$\beta$ large enough'' means 
``$\beta>\beta_0$ with some constant $\beta_0$
depending also only on $\Gamma$, $\Lambda$, $c$, and $\varphi$''. 

The following theorem shows that the breaking of linearized rotational 
symmetry $w\leadsto w+w_{\rm{const}}$ induced by the boundary conditions
persists in the thermodynamic limit $E\uparrow\Lambda$, provided that 
$\beta$ is large enough. 

\begin{theorem}{\bf (Spontaneous breaking of linearized rotational symmetry).}
\label{thm:main-result}
\noindent There is a constant $\celf>0$ such that for all $\beta$
large enough and for all $t\in\R$, 
\begin{align}
\label{eq:lower-bound-fourier-trafo-main}
\inf_{E\Subset\Lambda}\inf_{x,y\in\R^3}\min_{i,j\in[3]}
\EE_{P_{\beta,E}}\Big[\ee^{\ii t(w^*_{ij}(x,I)-w^*_{ij}(y,I))}\Big]
\ge \exp\Big\{-\frac{t^2}2\,\ee^{-\celf\beta}\Big\},  
\end{align}
and consequently
\begin{align}
\sup_{E\Subset\Lambda}\sup_{x,y\in\R^3}\max_{i,j\in[3]}
\var_{P_{\beta,E}}\big(w^*_{ij}(x,I)-w^*_{ij}(y,I)\big) \le \ee^{-\celf\beta}.
\end{align}
\end{theorem}
We remark that the symmetry $I\leftrightarrow -I$ implies that 
$w^*(x,I)$ is a centered random matrix. 
Since $w^*(x,I)$ encodes in particular the orientation of the crystal 
at location $x$, this result may be interpreted as the presence of 
long range orientational order in the thermodynamic limit.

\paragraph{Discussion of the model.} 
If we compare our model to the rotator model, 
purely elastic deformations correspond to ``spin wave'' contributions, while 
deformations induced by the Burgers vectors correspond to ``vortex'' contributions. 
In our model, the purely elastic deformations are orthogonal to the deformations 
induced by the Burgers vectors in a suitable inner product $\scalara{{\cdot}}{{\cdot}}_F$; 
this is made precise in Eq.\ \eqref{eq:orthogonality} below. 
Therefore, we do not model the purely elastic part stochastically. 
It would not be relevant for our purposes, because in a linearized model, it is 
expected to be independent of the Burgers vectors anyway. 

The mixed continuum / lattice structure of the model has the following 
motivation: A realistic microscopic description of a crystal at positive
temperature would be very complicated, including, e.g., vacancies and
interstitial atoms. This makes a global indexing of all atoms by a lattice
intrinsically hard. On a mesoscopic scale, we expect that these difficulties
can be neglected when the elastic parameters of the model are renormalized.
Our model should be understood as a two-scale description of the crystal 
in which the microscopic Burgers vectors are represented on a scale such that
their discreteness is still visible, but the physical space is smoothed out;
recall that the factor $\epsilon$ used in the approximation~\eqref{e:F-first}
encodes the ratio between the two scales.
Hence we use continuous derivatives in physical space, but 
discrete calculus for the Burgers vectors.
The mesoscopic lattice $V_\Lambda$ serves only as a convenient
spatial regularization. Unlike the microscopic lattice $\Gamma$,
it has no intrinsic physical meaning.

\paragraph{Organization of the paper.}
In Section~\ref{sec:MinimizingEnergy}, we identify the minimal energy configuration 
$w^\ast$ in the sense of Proposition~\ref{pro:convergence-to-w-star} in the appropriate 
Sobolev space. 
Section~\ref{sec:ClusterExpansion} deals with the statistical mechanics of Burgers 
vector configurations by means of a Sine-Gordon transformation and a cluster expansion 
in the spirit of the Fr\"ohlich-Spencer treatment of the Villain model \cite{MR634447,MR649811}. 
Section~\ref{sec:bounding-obs} provides the bounds for the observable, which manifests 
the spontaneous breaking of linearized rotational symmetry. It uses variants of a dipole 
expansion, which we provide in the appendix.

\section{Minimizing the elastic energy}
\label{sec:MinimizingEnergy}
In this section, we collect various properties of $H_{\rm el}^*(I)$ defined in \eqref{eq:Hel-star-def}.
In particular, we prove Proposition~\ref{pro:convergence-to-w-star}.

\subsection{Sobolev spaces}
Let $\bbV$ be a finite-dimensional $\C$-vector space with a norm $|\cdot|$
coming from a scalar product $\scalara{\cdot}{\cdot}_{\bbV}$. 
For integrable $f:\R^3\to\bbV$, let 
\begin{align}
\label{eq:def-Fouriertrafo}
\hat f(k)=(2\pi)^{-\frac32}\int_{\R^3}\ee^{-\ii\scalara{k}{x}}f(x)\,dx
\end{align}
denote its Fourier transform, 
normalized such that the transformation becomes unitary.
For any $\alpha\in\R$ and $f\in C^\infty_c(\R^3,\bbV)$, 
we define 
\begin{align}
\| f\|^\ve_\alpha:=\|\hat f\|_{2,\alpha}, \text{ where }
\|g\|_{2,\alpha}^2:=\int_{\R^3}  |k|^{2\alpha}|g(k)|^2\, dk.
\end{align}
We set
\begin{align}
C_\alpha(\bbV):=\big\{f\in C^\infty_c(\R^3,\bbV):\;\| f\|^\ve_\alpha<\infty\big\}.
\end{align}
Then $\| f\|^\ve_\alpha$ is a norm on the $\C$-vector space $C_\alpha(\bbV)$.
For $\alpha>-3/2$, we know that $C_\alpha(\bbV)=C^\infty_c(\R^3,\bbV)$
because $|k|^{2\alpha}$
is integrable near 0 and $\hat f$ decays fast at infinity.
Let $(L^{2\ve}_\alpha(\bbV),\|{\cdot}\|^\vee_\alpha)$ denote the completion of 
$(C_\alpha(\bbV),\|{\cdot}\|^\vee_\alpha)$ and
$L^2_\alpha(\bbV):=
\big\{g:\R^3\to\bbV\text{ measurable mod changes on null sets}: 
\|g\|_{2,\alpha}<\infty\big\}$.
The Fourier transform $f\mapsto\hat f$ gives rise to 
a natural isometric isomorphism $L^{2\ve}_\alpha(\bbV)\to L^2_\alpha(\bbV)$. 
For any $\alpha\in\R$, the sesquilinear form
\begin{align}
\scalara{\cdot}{\cdot}:C_{-\alpha}(\bbV)\times C_\alpha(\bbV)\to\C
,\quad
\scalara{f}{g}=\int_{\R^3}\scalara{f(x)}{g(x)}_{\bbV}\,dx
\end{align}
extends to a continuous sesquilinear form
\begin{align}
\scalara{\cdot}{\cdot}:L^{2\ve}_{-\alpha}(\bbV)\times L^{2\ve}_\alpha(\bbV)\to\C.
\end{align}
Partial derivatives $\partial_j:C^\infty_c(\R^3,\bbV)\to C^\infty_c(\R^3,\bbV)$ (acting component-wise)
extend to bounded operators $\partial_j:L^{2\ve}_{\alpha+1}(\bbV)\to L^{2\ve}_\alpha(\bbV)$. 
Consequently, the (component-wise) Laplace operator $\Delta=\sum_{j\in[3]}\partial_j^2$ 
extends to an isometric isomorphism
$\Delta:L^{2\ve}_{\alpha+2}(\bbV)\to L^{2\ve}_\alpha(\bbV)$. 

We will mainly have 
$\bbV=\bbV_j$ for $j\in\N_0$, where
\begin{align}
\bbV_j=\big\{ (a_{i_1\ldots i_jk})_{i_1,\ldots, i_j,k}\in\C^{3\times\cdots\times 3}:
a_{i_1\ldots i_jk}\text{ is antisymmetric in }i_1,\ldots, i_j\big\}, 
\end{align}
endowed with the norm 
\begin{align}
|a|=\left(\frac{1}{j!}\sum_{i_1,\ldots,i_j,k\in[3]}|a_{i_1\ldots i_jk}|^2\right)^{\frac12}. 
\end{align}
Functions with values in $\bbV_j$ are 
just $\C^3$-valued $j$-forms. Note that the last index $k$ has a special role 
since there is no antisymmetry condition for it. The real part of the space 
$\bbV_0=\C^3$ may be interpreted as a vector space containing
Burgers vectors.  
For any $\alpha\in\R$ and $j\in\N_0$, we introduce the exterior derivative 
$d_j:L^{2\ve}_{\alpha+1}(\bbV_j)\to L^{2\ve}_\alpha(\bbV_{j+1})$ and co-derivative 
$d_j^*:L^{2\ve}_{\alpha+1}(\bbV_{j+1})\to L^{2\ve}_\alpha(\bbV_j)$,
\begin{align}
(d_ja)_{i_1\ldots i_{j+1}k}&=
\sum_{l=1}^{j+1}
(-1)^{l+1}\partial_{i_l}a_{i_1\ldots \not i_l\ldots i_{j+1}k},
\\
(d_j^*a)_{i_1\ldots i_jk}&=-\sum_{i_0=1}^3
\partial_{i_0}a_{i_0i_1\ldots i_jk}.
\end{align}
They are adjoint to each other in the sense that for any $\alpha\in\R$,
\begin{align}
\scalara{d_j^*a}{b}=\scalara{a}{d_jb}\quad\text{ for }a\in L^{2\ve}_{-\alpha+1}(\bbV_{j+1}),\;
b\in L^{2\ve}_\alpha(\bbV_j).
\end{align}
Since in the following we are mostly interested in the cases $j=0,1,2$, we spell
out the definition of $d_j$ explicitly:
\begin{align}
(d_0f)_{ij}&=\partial_if_j,
\\
(d_1w)_{ijk}&=\partial_iw_{jk}-\partial_jw_{ik},
\\
(d_2b)_{ijkl}&=\partial_ib_{jkl}+\partial_jb_{kil}+\partial_kb_{ijl}
=
\partial_ib_{jkl}-\partial_jb_{ikl}+\partial_kb_{ijl}
.
\end{align}
The Laplace operator $\Delta:L^{2\ve}_{\alpha+2}(\bbV_j)\to L^{2\ve}_\alpha(\bbV_j)$ 
then fulfills
\begin{align}
\label{eq:def-laplace}
\Delta=-(d_j^*d_j+d_{j-1}d_{j-1}^*),
\quad (j\in\N_0);
\end{align}
here, $d_{-1}$ and $d_{-1}^*$ should be interpreted as 0. 
Equation~\eqref{eq:def-laplace} is only important for $j=0,1,2,3$ because $\bbV_j=\{0\}$
holds for $j\ge 4$ in three dimensions. 
To see \eqref{eq:def-laplace}, we calculate
\begin{align}
  (-d_j^*d_ja)_{i_1\ldots i_jk}
&= \sum_{i_0=1}^3\partial_{i_0}(d_ja)_{i_0i_1\ldots i_jk}=\sum_{i_0=1}^3\partial_{i_0}\left[\partial_{i_0}a_{i_1\ldots i_jk}
+\sum_{l=1}^j(-1)^{l+2}\partial_{i_l}a_{i_0i_1\ldots \not i_l\ldots i_jk}\right]
\nonumber\\
&=  \Delta a_{i_1\ldots i_jk}
-\sum_{i_0=1}^3\sum_{l=1}^j(-1)^{l+1}\partial_{i_0}\partial_{i_l}a_{i_0i_1\ldots \not i_l\ldots i_jk}
\end{align}
and
\begin{align}
   (-d_{j-1}d_{j-1}^*a)_{i_1\ldots i_jk}
=& -\sum_{l=1}^j(-1)^{l+1}\partial_{i_l}(d_{j-1}^*a)_{i_1\ldots \not i_l\ldots i_jk}
=\sum_{i_0=1}^3\sum_{l=1}^j(-1)^{l+1}\partial_{i_l}\partial_{i_0}
a_{i_0i_1\ldots \not i_l\ldots i_jk}.
\end{align}

\subsection{Elastic Hamiltonian}
Given $I\in\cI$ and $b=b(I)$, we calculate now $H^*_{\rm el}(I)$.
To begin with, we observe that $H_{\rm el}$, introduced in \eqref{eq:def_H_el}, is a quadratic form, and therefore it can be written as 
\begin{align}
\label{eq:H-el-with-scalar-product}
H_{\rm el}(w)=\scalara{w}{w}_F, 
\end{align}
with a sesquilinear form  $\scalara{{\cdot}}{{\cdot}}_F$ depending on $F$ defined in \eqref{eq:def-F}. More precisely, using $\scalara{A}{B}=\tr(AB^t)$ for 
the Euclidean scalar product for matrices $A,B$, we introduce 
 $\scalara{{\cdot}}{{\cdot}}_F:L^2(\R^3,\C^{3\times 3})\times L^2(\R^3,\C^{3\times 3})\to\C$ through 
\begin{align}
\scalara{w}{\tilde w}_F:=&\,
\int_{\R^3}\bigg[\frac{\lambda}{2}(\tr(\overline{w(x)+w^t(x)})\tr(\tilde w(x)+\tilde w^t(x))+\mu \tr\big[(\overline{w(x)+w^t(x)})(\tilde w(x)+\tilde w^t(x))\big]\bigg]\,dx
\nonumber\\
=&\,
2\sum_{i,j=1}^3\int_{\R^3}\left[\lambda(\overline{w_{ii}(x)}\tilde w_{jj}(x))
+\mu 
\overline{(w_{ij}(x)+w_{ji}(x))}\tilde w_{ij}(x)\right]\,dx.
\end{align}
Because of the stability condition \eqref{eq:def-F}, the inner product 
$\scalara{{\cdot}}{{\cdot}}_F$ is positive semidefinite.
For any $\alpha\in\R$, we consider the restriction of $\scalara{{\cdot}}{{\cdot}}_F$ to $C_{-\alpha}(\bbV_1)\times C_{\alpha}(\bbV_1)$, and then extend it to a sesquilinear form 
$\scalara{{\cdot}}{{\cdot}}_F:L^{2\vee}_{-\alpha}(\bbV_1)\times L^{2\vee}_\alpha(\bbV_1)\to \C$ being continuous w.r.t.\ $\|{\cdot}\|_{-\alpha}^\vee$ and $\|{\cdot}\|_\alpha^\vee$.

In  \eqref{eq:Hel-star-def} we may take the infimum over $w^b+\ker(d_1) $ with a suitable  $w^b\in L^{2\vee}_{0}(\bbV_1) = L^2(\R^3,\C^{3\times 3})$ satisfying $d_1w^b=b(I)$.
We claim that a convenient choice is 
\begin{align}
\label{eq:def-wb}
w^b:=-\Delta^{-1}d_1^*b\in L^{2\vee}_0(\bbV_1).
\end{align}
To see this, we observe that $\Delta^{-1}$ commutes with $d_1$ and $d_1^*$
because $\Delta^{-1}$ corresponds to multiplication with the scalar $|k|^{-2}$
in Fourier space, and $d_1,d_1^*$ correspond to (multi-component) multiplication 
operators in Fourier space, as well. Therefore, since  
$b=b(I)\in\operatorname{ker}(d_2:L^{2\vee}_{-1}(\bbV_2)\to L^{2\vee}_{-2}(\bbV_3))$, 
we obtain 
\begin{align}
d_1w^b=-d_1d_1^*\Delta^{-1}b=-\Delta^{-1}d_1d_1^*b
=-\Delta^{-1}(d_1d_1^*+d_2^*d_2)b=b.
\end{align}
Using that $\operatorname{ker}(d_1:C_0(\bbV_1)\to C_{-1}(\bbV_2))$ is dense in 
$\operatorname{ker}(d_1:L^{2\vee}_0(\bbV_1)\to L^{2\vee}_{-1}(\bbV_2)) =
\operatorname{range}(d_0:L^{2\vee}_1(\bbV_0)\to L^{2\vee}_{0}(\bbV_1))$, 
we obtain
\begin{align}
H_{\rm el}^*(I)
=&\inf_{\substack{w\in L^{2\vee}_0(\bbV_1):\\d_1w=0}}\scalara{w^b+w}{w^b+w}_F
=
   \inf_{\psi\in L^{2\vee}_1(\bbV_0) }\scalara{w^b+d_0\psi}{w^b+d_0\psi}_F.
   \label{eq:inf-H-ell}
\end{align}

\subsection{Minimizer}
\paragraph{Differential operators.}
In order to analyze the $d_0\psi$-dependence in \eqref{eq:inf-H-ell}, we derive an adjoint $\nabla^F$ for $d_0$ with respect to $\scalara{{\cdot}}{{\cdot}}_F$ and $\scalara{{\cdot}}{{\cdot}}$. 
Let $\nabla^F:L^{2\ve}_{-\alpha+1}(\bbV_1)\to L^{2\ve}_{-\alpha}(\bbV_0)$,
\begin{align}
(\nabla^Fg)_j:=-2\sum_{i=1}^3
\big[\lambda \partial_jg_{ii}
+\mu 
\partial_i(g_{ij}+g_{ji})\big],\quad (j\in[3]).
\end{align}
Indeed, it satisfies the following adjointness relation for any $g\in L^{2\vee}_{-\alpha}(\bbV_1)$ and $f\in L^{2\vee}_{\alpha+1}(\bbV_0)$ 
(using that $-\partial_j$ is adjoint 
to $\partial_j$ w.r.t.\ $\scalara{{\cdot}}{{\cdot}}$):
\begin{align}
&\scalara{g}{d_0f}_F
=
2\sum_{i,j=1}^3\big[\lambda\scalara{g_{ii}}
{\partial_jf_j}
+\mu 
\scalara{g_{ij}+g_{ji}}{\partial_if_j}\big]
\nonumber
\\&=
-2\sum_{i,j=1}^3
\scalara{\lambda \partial_jg_{ii}
+\mu 
\partial_i(g_{ij}+g_{ji})}{f_j}
=
\scalara{\nabla^Fg}{f}.
\label{eq:ScalarF-Scalar-Correspondence}
\end{align}
The identity $\scalara{d_0\psi}{d_0 \psi}_F=\scalara{\nabla^Fd_0\psi}{\psi}$ motivates us to introduce the following differential operator for any $\alpha\in\R$:
\begin{align}
&D:=\tfrac12\nabla^Fd_0:
L^{2\vee}_{\alpha+2}(\bbV_0)\to L^{2\vee}_\alpha(\bbV_0),
\\
&D\psi=-\mu\Delta\psi-(\mu+\lambda)\operatorname{grad}\operatorname{div}\psi
=\Big(-\mu\Delta\psi_j-(\mu+\lambda)
\partial_j\sum_{i=1}^3\partial_i\psi_i\Big)_{j\in[3]}.
\end{align}
At this moment, we are most interested in the case $\alpha=-1$; the case of general values for $\alpha$ is needed for regularity considerations in the proof of Proposition~\ref{pro:convergence-to-w-star} later on.

\begin{lemma}[Properties of $D$]
For any $\alpha\in\R$,
the map $D:L^{2\vee}_{\alpha+2}(\bbV_0)\to L^{2\vee}_\alpha(\bbV_0)$ is invertible
with the inverse
$D^{-1}:L^{2\vee}_\alpha(\bbV_0)\to L^{2\vee}_{\alpha+2}(\bbV_0)$,
\begin{align}
D^{-1}\psi = -\Delta^{-1}
\left(\frac{1}{\mu}\psi+
\left(\frac{1}{2\mu+\lambda}-\frac{1}{\mu}\right)
\Delta^{-1}\operatorname{grad}\operatorname{div}\psi\right).
\end{align}
In coordinate notation,
\begin{align}
(D^{-1}\psi)_j = -\Delta^{-1}
\left(\frac{1}{\mu}\psi_j+
\left(\frac{1}{2\mu+\lambda}-\frac{1}{\mu}\right)
\Delta^{-1}\partial_j\sum_{i=1}^3\partial_i\psi_i\right).
\end{align}
The map $D$ is symmetric for $\alpha=-1$, i.e., 
$\scalara{\tilde\psi}{D\psi}=\scalara{D\tilde\psi}{\psi}$
for $\psi,\tilde\psi\in L^{2\vee}_1(\bbV_0)$,
and bounded from above and from below as follows:
\begin{align}
\label{eq:comparison-Delta-D}
0\le \mu\scalara{\psi}{-\Delta\psi}
\le \scalara{\psi}{D\psi}
\le (2\mu+\lambda)\scalara{\psi}{-\Delta\psi}.
\end{align}
\end{lemma}
\begin{proof}
Using $\operatorname{div}\operatorname{grad}=\Delta$ and abbreviating
\begin{align}
\gamma:=\frac{1}{2\mu+\lambda}-\frac{1}{\mu},
\end{align}
we calculate
\begin{align}
D^{-1} D\psi=&
-\Delta^{-1}\left(\frac{1}{\mu}\operatorname{Id}+
\gamma
\Delta^{-1}\operatorname{grad}\operatorname{div}\right)
\left(-\mu\Delta\psi
-(\mu+\lambda)\operatorname{grad}\operatorname{div}\psi\right)
\nonumber
\\=&
\,\psi
+
  \gamma  \mu
\Delta^{-1}\operatorname{grad}\operatorname{div}\psi
+\frac{\mu+\lambda}{\mu}\Delta^{-1}
\operatorname{grad}\operatorname{div}\psi
+\gamma(\mu+\lambda)
\Delta^{-1}\operatorname{grad}\operatorname{div}\psi
=\psi
\end{align}
and similarly $DD^{-1}=\operatorname{id}$.

By the adjointness property \eqref{eq:ScalarF-Scalar-Correspondence}, the symmetry of $D$ is obvious by its definition:
$2\scalara{D\psi'}{\psi}
=\scalara{\nabla^Fd_0\psi'}{\psi}
=\scalara{d_0\psi'}{d_0 \psi}_F$ for $\psi,\psi'\in L^{2\vee}_1(\bbV_0)$.
Furthermore, one has 
\begin{align}
\scalara{\psi}{D\psi}=
\mu\scalara{\psi}{-\Delta\psi}+
(\mu+\lambda)\scalara{\operatorname{div}\psi}{\operatorname{div}\psi}.
\end{align}
We claim that
\begin{align}
\label{eq:fact-div-psi}
\scalara{\operatorname{div}\psi}{\operatorname{div}\psi}
\le
\scalara{\psi}{-\Delta\psi}.
\end{align}
This is best seen using a Fourier transform and the Cauchy-Schwarz inequality in $\C^3$:
\begin{align}
  \scalara{\operatorname{div}\psi}{\operatorname{div}\psi}
=&\|\wh{\operatorname{div}\psi}\|^2_2
=\|\ii k\cdot\wh{\psi}(k)\|^2_2
\le \||k| |\wh{\psi}(k)| \|^2_2
=\|\wh{d_0\psi}\|^2_2
=\|d_0\psi\|^2_2
=\scalara{\psi}{-\Delta\psi}
\end{align}
where $k\cdot\wh{\psi}(k)$ denotes the Euclidean scalar product in $\C^3$. 
Using fact \eqref{eq:fact-div-psi} and the stability condition
for $\mu$ and $\lambda$ given in \eqref{eq:def-F}, 
which implies $\mu+\lambda>0$,
we obtain also claim \eqref{eq:comparison-Delta-D}.
\end{proof}

\paragraph{Definition of the minimizer.}
In the next lemma, it is shown that the minimizer of the elastic
energy has the following form:
\begin{align}
\label{eq:def-w-star}
w^*&:=w^b+d_0\psi^*
\end{align}
with $w^b$ defined in \eqref{eq:def-wb}, 
\begin{align}
\label{eq:def-v-b}
\psi^*&:=D^{-1}v^b,\quad\text{and}\quad
v^b :=-\frac12\nabla^F w^b.
\end{align}

\begin{lemma}[Minimizer of the elastic energy]
\label{le:minimizer-elastic-energy}
The infimum in \eqref{eq:inf-H-ell} is a minimum:
\begin{align}
\label{eq:claim-minimizer}
H_{\rm el}^*(I)=\scalara{w^*}{w^*}_F.
\end{align}
It is unique in the following sense:
For all $w\in L^{2\vee}_0(\bbV_1)$ with $d_1w=b(I)$ and
$\scalara{w}{w}_F=\scalara{w^*}{w^*}_F$,
we have $w=w^*$. The summands of the minimizer $w^*$ given in \eqref{eq:def-w-star} have the following components:
\begin{align}
\label{eq:components-of-w-b-componentwise}
w^b_{ij}&=-(d_1^*\Delta^{-1}b)_{ij}=\sum_{l=1}^3\Delta^{-1}\partial_l b_{lij}, \\
d_0\psi^*_{ij}&=\Delta^{-1}\partial_i\sum_{k=1}^3\left(\partial_k(d_1^*\Delta^{-1}b)_{jk}
+\frac{\lambda}{2\mu+\lambda}\partial_j(d_1^*\Delta^{-1}b)_{kk}\right)\nonumber\\
&=-\partial_i\Delta^{-2}\sum_{k,l=1}^3 
\left(
\partial_k\partial_lb_{ljk}+
\frac{\lambda}{2\mu+\lambda}
\partial_j\partial_lb_{lkk}
\right), \qquad i,j\in[3].
\label{eq:components-of-d0-psi-stern-componentwise}
\end{align}
\end{lemma}
\begin{proof}
The calculation
$\nabla^F(w^b+d_0\psi^*)=
\nabla^F(w^b-\frac12d_0D^{-1}\nabla^F w^b)
=
\nabla^Fw^b-DD^{-1}\nabla^F w^b=0
$
shows that
the function $\psi^*$
solves the system of equations
\begin{align}
\label{eq:orthogonality-dual}
\scalara{\nabla^F(w^b+d_0\psi^*)}{f}=0, \quad (f\in L^{2\vee}_1(\bbV_0)),
\end{align}
or equivalently, using \eqref{eq:ScalarF-Scalar-Correspondence} and 
\eqref{eq:def-w-star},
\begin{align}
\label{eq:orthogonality}
\scalara{w^*}{d_0f}_F=0, \quad (f\in L^{2\vee}_1(\bbV_0)).
\end{align}
By the above, the following calculation shows that $w^*$ is a minimizer 
in \eqref{eq:inf-H-ell}
as claimed:
For all $f\in L^{2\vee}_1(\bbV_0)$:
\begin{align}
\scalara{w^*+d_0f}{w^*+d_0f}_F&=
\scalara{w^*}{w^*}_F+
2\operatorname{Re}\scalara{w^*}{d_0f}_F+
\scalara{d_0f}{d_0 f}_F
\nonumber\\&
=
\scalara{w^*}{w^*}_F+
\scalara{d_0f}{d_0 f}_F
\ge \scalara{w^*}{w^*}_F.
\end{align}
Furthermore, using \eqref{eq:comparison-Delta-D} we obtain :
\begin{align}
\scalara{d_0f}{d_0 f}_F
=&\scalara{\nabla^Fd_0f}{f}
=2\scalara{Df}{f}
\ge2\mu\scalara{f}{-\Delta f}=2\mu\scalara{d_0f}{d_0f}
=2\mu\|d_0f\|_2^2.
\end{align}
In particular, $d_0f\neq 0$ implies $\scalara{d_0f}{d_0 f}_F>0$,
which yields the claimed uniqueness of the minimizer.
Let $i,j\in[3]$. The identity  \eqref{eq:components-of-w-b-componentwise} follows from the definition \eqref{eq:def-wb} of $w^b$. 
Using it, we express $v^b \in L^{2\vee}_{-1}(\bbV_0)$ as follows:
\begin{align}
v^b_j&=\sum_{k,l=1}^3 \big[\lambda\partial_j\partial_l(\Delta^{-1}b)_{lkk}
+\mu \partial_k\partial_l((\Delta^{-1}b)_{lkj}+(\Delta^{-1}b)_{ljk})\big]
\nonumber\\
&=\Delta^{-1}\sum_{k,l=1}^3 \big[\lambda\partial_j\partial_lb_{lkk}
+\mu \partial_k\partial_lb_{ljk}\big].
\end{align}
Because of the antisymmetry 
$\partial_k\partial_l(\Delta^{-1}b)_{lkj}
=-\partial_l\partial_k(\Delta^{-1}b)_{klj}$,
one term dropped out in the last step.
It follows that 
\begin{align}
\psi^*_j&=(D^{-1}v^b)_j
=
-\Delta^{-1}
\left(\frac{1}{\mu}v^b_j+
\left(\frac{1}{2\mu+\lambda}-\frac{1}{\mu}\right)
\Delta^{-1}\partial_j\sum_{m=1}^3\partial_mv^b_m\right)
\\&=
-\Delta^{-2}\sum_{k,l=1}^3 
\Bigg(\frac{1}{\mu}[\lambda\partial_j\partial_lb_{lkk}
     +\mu \partial_k\partial_lb_{ljk}]+
\left(\frac{1}{2\mu+\lambda}-\frac{1}{\mu}\Bigg)
\Delta^{-1}\partial_j\sum_{m=1}^3\partial_m
[\lambda\partial_m\partial_lb_{lkk}
+\mu \partial_k\partial_lb_{lmk}]
\right). \nonumber
\end{align}
Using $\sum_{l,m}\partial_m\partial_lb_{lmk}=0$ from 
the antisymmetry $b_{lmk}=-b_{mlk}$, this equals
\begin{align}
\psi^*_j&=
-\Delta^{-2}\sum_{k,l=1}^3 
\left(\frac{1}{\mu}[\lambda\partial_j\partial_lb_{lkk}
+\mu \partial_k\partial_lb_{ljk}]+
\left(\frac{1}{2\mu+\lambda}-\frac{1}{\mu}\right)
\lambda\partial_j\partial_lb_{lkk}
\right)
\nonumber\\&=
-\Delta^{-2}\sum_{k,l=1}^3 
\left(
\partial_k\partial_lb_{ljk}+
\frac{\lambda}{2\mu+\lambda}
\partial_j\partial_lb_{lkk}
\right).
\end{align}
This shows that $d_0\psi^*$ has the form given in \eqref{eq:components-of-d0-psi-stern-componentwise}. 
\end{proof}

\medskip\noindent
\paragraph{Regularity of the minimizer.}
\begin{proof}[Proof of Proposition~\ref{pro:convergence-to-w-star}]
We set $L_{>\alpha}^{2\vee}(\bbV):=\bigcap_{\alpha':\alpha'>\alpha}L_{\alpha'}^{2\vee}(\bbV)$.
From 
$b(I)\in C^\infty_c(\R^3,\bbV_2)=\bigcap_{\alpha>-3/2}C_\alpha(\bbV_2)$ $\subseteq
L_{>-3/2}^{2\vee}(\bbV_2)$ it follows from \eqref{eq:def-wb} that 
$w^b\in L_{>-1/2}^{2\vee}(\bbV_1)$. Hence, by \eqref{eq:def-v-b}, 
$v^b\in L_{>-3/2}^{2\vee}(\bbV_0)$,
and then $\psi^*=D^{-1}v^b\in L_{>1/2}^{2\vee}(\bbV_0)$. We conclude 
$w^*=w^b+d_0\psi^*\in L_{>-1/2}^{2\vee}(\bbV_1)$. By Sobolev's embedding
theorem, $w^*$ is a bounded smooth function with all derivatives being 
bounded. 
In particular, pointwise evaluation $w^*(x)$ of $w^*$ makes sense for every
$x\in\R^3$.

For the remaining claim,
take a sequence $f^n\in L_1^{2\vee}(\bbV_0)$, $n\in\N$, with 
$H_{\rm{el}}(w^*+d_0f^n)\to H_{\rm{el}}(w^*)=H_{\rm{el}}^*(I)$ as $n\to\infty$.
Using $\operatorname{ker}(d_1:L^{2\vee}_0(\bbV_1)\to L^{2\vee}_{-1}(\bbV_2)) =
\operatorname{range}(d_0:L^{2\vee}_1(\bbV_0)\to L^{2\vee}_{0}(\bbV_1))$, 
it suffices to show that $\|d_0f^n\|_2$ converges to $0$
as $n\to\infty$.
In view of the system \eqref{eq:orthogonality} of equations, we know 
\begin{align}
2\scalara{f^n}{Df^n}&=\scalara{d_0f^n}{d_0f^n}_F
=\scalara{d_0f^n}{d_0f^n}_F+2\real\scalara{w^*}{d_0f^n}_F \nonumber\\
&=\scalara{w^*+d_0f^n}{w^*+d_0f^n}_F-\scalara{w^*}{w^*}_F
\nonumber\\
&=H_{\rm{el}}(w^*+d_0f^n)-H_{\rm{el}}(w^*)
\stackrel{n\to\infty}{\longrightarrow}0.
\end{align}
Using the comparison \eqref{eq:comparison-Delta-D} between $D$ and
$-\Delta$,
we conclude
\begin{align}
\|d_0f^n\|_2^2=\scalara{d_0f^n}{d_0f^n}=\scalara{f^n}{-\Delta f^n}
\stackrel{n\to\infty}{\longrightarrow}0.
\end{align}
\end{proof}

We remark that the facts $d_2b(I)=0$ and $b(I)\in C^\infty_c(\R^3,\bbV_2)$
imply $\int_{\R^3}b(I)(x)\,dx=0$ and hence 
$b(I)\in C_\alpha(\bbV_2)$ for all $\alpha>-5/2$, not only for all $\alpha>-3/2$.
As a consequence, $w^*\in L_{>-3/2}^{2\vee}(\bbV_1)$. 

\section{Cluster expansion}
\label{sec:ClusterExpansion}

We now develop a cluster expansion (polymer expansion) of the measures 
$P_{\beta,E}$
  defined in \eqref{e:Pbeta-def}, using the strategy of Fr\"ohlich and Spencer \cite{MR649811}.
In the following, $E\Subset\Lambda$ is a given finite set 
of edges in the mesoscopic lattice. We take the thermodynamic
limit $E\uparrow\Lambda$ only in the end.
\subsection{Sine-Gordon transformation}
The elastic energy $H^*_{\rm el}(I)$ defined in \eqref{eq:Hel-star-def} is a 
quadratic form. If $I=I_1+\cdots+I_n$ is the decomposition of $I$ into its connected
components, the mixed terms in $H^*_{\rm el}(I)$ induce non-local interactions 
between different $I_i$ and $I_j$. The Sine-Gordon transformation introduced now 
is a tool to avoid these non-localities. 

Because the quadratic form $H^*_{\rm el}$ is positive semidefinite, the 
function $\exp\{-\beta H^*_{\rm el}\}$ is the Fourier transform of a centered 
Gaussian random vector 
$\phi=(\phi_e)_{e\in E}$ on some auxiliary probability space with corresponding 
expectation operator denoted by $\bbE$:
\begin{align}
\bbE\big[\ee^{\ii\scalara{\phi}{I}}\big]=\ee^{-\beta H^*_{\rm el}(I)}.
\end{align}
For any observable of the form $\cI\ni I\mapsto\scalara{\sigma}{I}$ with 
$\sigma\in\R^E$, we define
\begin{align}
\label{eq:def-cal-Z}
\cZ_{\beta,\phi}:=&\sum_{I\in\cI}\ee^{\ii\scalara{\phi}{I}}\ee^{-\beta H_{\rm disl}(I)}, \\
Z_\beta(\sigma):= & \sum_{I\in\cI} \ee^{\ii\scalara{\sigma}{I}}\ee^{-\beta (H^*_{\rm el}(I)+H_{\rm disl}(I))}
=\bbE\left[\cZ_{\beta,\sigma+\phi}\right].
\label{eq:def-Z-beta-sigma}
\end{align}
In order to exchange expectation and summation,
we used that $\ee^{-\beta H_{\rm disl}(I)}$ is summable over the 
set $\cI$ by \eqref{eq:lower-bound-H-disl-in-terms-of-I}.
Note that $Z_\beta=Z_\beta(0)$ implies
\begin{align}
\label{eq:ratio-Zsigma-Z0}
  \frac{Z_\beta(\sigma)}{Z_\beta(0)}=\EE_{P_\beta}\big[\ee^{\ii\scalara{\sigma}{I}}\big].
\end{align}

\subsection{Preliminaries on cluster expansions}
In this section, we collect some background on cluster expansions (polymer expansions).
For recent treatments of cluster expansions, see in particular Poghosyan and Ueltschi \cite{MR2531305} or 
Bovier and Zahradn\'{i}k \cite{MR1788485} and references. To make our presentation most accessible, we use the textbook version given in \cite{friedli-velenik2018}.

Let $\cB$ denote the set of all non-empty connected subsets of $E$.
We call $X,Y\in\cB$ compatible, $X\not\sim Y$, if no edge in $X$ has a 
common vertex with an edge in $Y$. 
Otherwise $X,Y$ are called incompatible, $X\sim Y$. In particular, 
$X\sim X$. Recall $\supp I=\{e\in E:I_e\neq 0\}$ for $I\in\cI$, where $\cI$ is 
defined in \eqref{eq:def-cal-I}. Let 
\begin{align}
 \label{eq:def-cJ}
\cJ=\{I\in\cI:\supp I\in\cB\}.
\end{align} 
The incompatibility relation $\sim$ on
$\cB$ is inherited to an incompatibility relation, also denoted by $\sim$,
on $\cJ$ via
\begin{align}
I\sim I'\quad:\Leftrightarrow\quad
\supp I\sim \supp I'.
\end{align}
Every subset of $E$ can be uniquely decomposed in a set of pairwise compatible 
connected components, which is a subset of $\cB$. For $n\in\N$, let 
\begin{align}
\label{eq:def-cJnotsim}
\cJ_{\not\sim}^n=
& \{(I_1,\ldots,I_n)\in\cJ^n:I_i\not\sim I_j\text{ for all }i\neq j\}.
\end{align}
Consider $I\in\cI$ and the connected components $X_1,\ldots,X_n$ ($n\in\N_0$) 
of $\supp I$.
We set $I_j:=I1_{X_j}\in\cJ$. Here it is crucial that the Kirchhoff rule
\eqref{eq:Kirchhoff} holds for $I$ if and only if it holds for all $I_j$. 
Then, using the locality of $H_{\rm disl}$ given in Assumption 
\ref{ass:H-disl}, we obtain
\begin{align}
\scalara{\phi}{I}=\sum_{j=1}^n\scalara{\phi}{I_j},\qquad
H_{\rm disl}(I)=\sum_{j=1}^n H_{\rm disl}(I_j).
\end{align}
For $I\in\cI$ and some $\beta>0$, we abbreviate
\begin{align}
\label{eq:KIphi}
K(I,\phi):=
\ee^{\ii\scalara{\phi}{I}}\ee^{-\beta H_{\rm disl}(I)}. 
\end{align}
The function $K$ fulfills the following important factorization property: 
For $I\in\cI$ with connected components $I_1,\ldots,I_n$ as above, one has
\begin{align}
\label{eq:factorization-K-I-phi}
K(I,\phi)=\prod_{j=1}^nK(I_j,\phi).
\end{align}
This fact relies on the dimension being at least $3$.
In $d=2$, the Burgers vector density would not be locally neutral,
resulting in a significant complication of the argument (as in
\cite{MR634447} compared to \cite{MR649811}).
In view of the definition~\eqref{eq:def-cal-Z} of $\cZ_{\beta,\phi}$, 
equation \eqref{eq:factorization-K-I-phi} yields
\begin{align}
\label{eq:Zsigmaxy} 
\cZ_{\beta,\phi}
=& \sum_{I\in\cI}K(I,\phi)
= 1+\sum_{n=1}^\infty\frac{1}{n!}\sum_{(I_1,\ldots,I_n)\in\cJ^n_{\not\sim}}
\prod_{j=1}^n K(I_j,\phi).
\end{align}
The summand $1$ comes from the contribution of $I=0$, using $H_{\rm disl}(0)=0$.
Recall that by \eqref{eq:lower-bound-H-disl-in-terms-of-I} 
$|K(I,\phi)|=\ee^{-\beta H_{\rm disl}(I)}\le 
\ee^{-\beta c\|I\|_1}$ is summable over $I\in\cI$, which shows that
all expressions in \eqref{eq:Zsigmaxy} are absolutely summable. 
To control $\cZ_{\beta,\phi}$, we use a cluster expansion.
Next we cite the relevant theorems. 

Let $\cG_n$ denote the set of all connected subgraphs $G_n=([n],E_n)$ of the 
complete graph with vertex set $[n]=\{1,\ldots,n\}$. Let 
$\cE_n=\{E_n:G_n=([n],E_n)\in\cG_n\}$ denote the set of all corresponding 
edge sets. Consider the Ursell functions
\begin{align}
U(I_1,\ldots,I_n)=\frac{1}{n!}\sum_{E_n\in\cE_n}\prod_{\{i,j\}\in E_n}(-1_{\{I_i\sim I_j\}}).
\end{align}

Let $\tilde \cJ$ be any finite set 
endowed with a reflexive and symmetric incompatibility relation $\sim$.
We define $\tilde\cJ_{\not\sim}^n$ by \eqref{eq:def-cJnotsim}
with $\cJ$ replaced by $\tilde\cJ$.

\begin{fact}[Formal cluster expansion, {\cite[Proposition 5.3]{friedli-velenik2018}}]
\label{fact:velenik-clusterexpansion}
For every $I\in\tilde\cJ$, let $K(I)$ be a variable. Consider the polynomial 
in these variables 
\begin{align}
{\rm Z}:=1+\sum_{n=1}^\infty\frac{1}{n!}\sum_{(I_1,\ldots,I_n)\in\tilde\cJ^n_{\not\sim}}
\prod_{j=1}^n K(I_j). 
\end{align}
As a formal power series 
\begin{align}
\label{eq:cluster-expansion}
\log{\rm Z}=\sum_{n=1}^\infty \sum_{(I_1,\ldots,I_n)\in\tilde\cJ^n}
U(I_1,\ldots,I_n)\prod_{j=1}^n K(I_j).
\end{align}
Moreover, if the right hand side in
\eqref{eq:cluster-expansion} is absolutely summable,
then equation \eqref{eq:cluster-expansion} holds also in the 
classical sense as follows:
$\exp(\operatorname{rhs} \eqref{eq:cluster-expansion})={\rm Z}$.
\end{fact}
A criterion for convergence of the cluster expansion is cited in
the following fact: 
\begin{fact}[Convergence of cluster expansions, {\cite[Theorem 5.4]{friedli-velenik2018}}]
\label{fact:velenik-conv-clusterexpansion}
Assume that there are ``sizes'' $(a(I))_{I\in\tilde\cJ}$ $\in\R_{\ge 0}^{\tilde\cJ}$
and ``weights'' $(K(I))_{I\in\tilde\cJ}\in\C^{\tilde\cJ}$
such that for all $I\in\tilde\cJ$, the following bound holds: 
\begin{align}
\label{eq:assumption-cluster}
\sum_{J\in \tilde\cJ} |K(J)|1_{\{I\sim J\}}\ee^{a(J)}\le a(I).
\end{align}
Then we have for all $J\in\tilde\cJ$:
\begin{align}
1+\sum_{n=2}^\infty n \sum_{(I_1,\ldots,I_{n-1})\in \tilde\cJ^{n-1}}
|U(J,I_1,\ldots,I_{n-1})|\prod_{j=1}^{n-1} |K(I_j)|
\le \ee^{a(J)}.
\end{align}
Moreover, in this case, the series \eqref{eq:cluster-expansion}
is absolutely convergent.
\end{fact}

\subsection{Partial partition sums}
We take a sequence $(\cJ_m)_{m\in\N}$ of finite subsets $\cJ_m\subseteq \cJ$
with $\cJ_m\uparrow \cJ$ and set $\cJ_\infty:=\cJ$, with $\cJ$ being defined in \eqref{eq:def-cJ}. 
For $m\in\N\cup\{\infty\}$ and $I\in\cI$, let
\begin{align}
\label{eq:zm-series}
z_m(\beta,I):=\sum_{n=1}^\infty\sum_{(I_1,\ldots,I_n)\in\cJ_m^n} 
U(I_1,\ldots,I_n) 1_{\{I_1+\cdots +I_n=I\}}
\prod_{j=1}^n\ee^{-\beta H_{\rm disl}(I_j)}\in\R
\end{align}
whenever this double series is absolutely convergent. 
Note that $z_m(\beta,I)=z_m(\beta,-I)$
because $H_{\rm disl}(I)=H_{\rm disl}(-I)$ by Assumption~\ref{ass:H-disl}. 
We abbreviate also $z(\beta,I):=z_\infty(\beta,I)$.
Uniformly in $m$, the summands in the series~\eqref{eq:zm-series}
are dominated by
the corresponding ones in $z^+(\beta,I):=z^+_\infty(\beta,I)$, where
\begin{align}
\label{eq:def-zmplus}
z^+_m(\beta,I)&:=\sum_{n=1}^\infty\sum_{(I_1,\ldots,I_n)\in\cJ_m^n} 
|U(I_1,\ldots,I_n)| 1_{\{I_1+\cdots +I_n=I\}}
\prod_{j=1}^n\ee^{-\beta H_{\rm disl}(I_j)}\in[0,\infty].
\end{align}
By monotone convergence for series, 
\begin{align}
\label{eq:monotone-conv-z+m}
z^+_m(\beta,I)\uparrow z^+_\infty(\beta,I)
\text{ as }m\to\infty.
\end{align}
For $I\in\cI$ we define its size 
\begin{align}
\size I :=\|I\|_1+\diam\supp I. 
\end{align}
Here $\diam$ denotes the diameter in the graph distance in the mesoscopic lattice $G=(V,E)$.
The size has the following properties. For $I_1,I_2\in\cI$ with $I_1\sim I_2$, one has 
\begin{align}
\label{eq:size-subadditive}
\size (I_1+I_2)\le\size I_1+\size I_2.
\end{align}
Recall that $I$ takes values in the microscopic lattice $\Gamma$. We set 
\begin{align}
\eta:=\min\{|\gamma|:\gamma\in\Gamma\setminus\{0\}\}
\end{align}
and observe for all $I\in \cI$:
\begin{align}
\label{eq:supp-vs-1-norm}
\eta|\supp I|\le \|I\|_1.
\end{align}
If in addition $\supp I$ is connected, we have $\diam\supp I\le |\supp I|$
and hence
\begin{align}
\label{eq:eq-1-norm-size}
\|I\|_1\le\size I\le\czwei\|I\|_1
\end{align}
with the constant $\czwei:=1+\eta^{-1}$. Using the constant $c$ from 
\eqref{eq:lower-bound-H-disl-in-terms-of-I}, let 
$\ceins=\ceins(c,\eta):=c/(2\czwei)$. Then, for $I\in\cI$, it follows 
\begin{align}
\label{eq:estimate-Hdisl-size}
H_{\rm disl}(I)\ge c\|I\|_1 \ge \frac{c}{\czwei}\size I \ge \ceins\size I.   
\end{align}
We choose now a constant 
$\cfuenf=\cfuenf(c,\eta)$ with 
$0<\cfuenf<\ceins$ and set $\cdrei=\cdrei(c,\eta):=\ceins-\cfuenf>0$. 
Fact~\ref{fact:velenik-conv-clusterexpansion} is applied twice, 
later with the weight $K(I,\phi)$ introduced in \eqref{eq:KIphi}, but
first with the weight 
\begin{align}
\label{eq:def-K}
K(J):=\ee^{-\beta\ceins\size J},\quad J\in\cJ,
\end{align}
and the size function $a:\cJ\to\R_{>0}$,
\begin{align}
a(J):=\beta \cfuenf \eta |\supp J|, \quad J\in\cJ.
\end{align}
The following lemma serves to verify the hypothesis
\eqref{eq:assumption-cluster} of the cluster expansion. 

\begin{lemma}[Peierls argument]
\label{le:hypothesis-convergence-cluster}
There is $\cvierundzwanzig>0$ such that for all $\beta$ large enough, 
\begin{align}
\label{eq:claim-Peierls}
\sup_{E\Subset\Lambda}\sup_{o\in E}\sum_{\substack{J\in\cJ:\\o\in\supp J}}
\ee^{-\beta\cdrei\size J}
\le \ee^{-\beta \cvierundzwanzig}.
\end{align}
Furthermore, one has 
\begin{align}
\label{eq:hilfsbeh-sum-K-exp-a}
\sup_{m\in\N}\sup_{E\Subset\Lambda}\sup_{o\in E}
\sum_{\substack{J\in \cJ_m:\\o\in\supp J}} |K(J)|\ee^{a(J)}
\le \ee^{-\beta \cvierundzwanzig},
\end{align}
and the hypothesis \eqref{eq:assumption-cluster} holds for $\tilde\cJ=\cJ_m$ for 
all $m\in\N$ and all $\beta$ large enough.
\end{lemma}
\begin{proof}
The claim \eqref{eq:claim-Peierls}
is verified as follows:
Take $o\in E\Subset\Lambda$. We estimate 
\begin{align}
\label{eq:remaining-claim}
\sum_{\substack{J\in\cJ:\\o\in\supp J}}
\ee^{-\beta\cdrei\size J}
\le 
\sum_{\substack{J\in\cJ:\\o\in\supp J}}
\ee^{-\beta\cdrei\|J\|_1}
=
\sum_{\substack{X\in\cB:\\o\in X}}
\sum_{\substack{J\in\cJ:\\\supp J=X}}
\ee^{-\beta\cdrei\|J\|_1}.
\end{align}
Dropping the condition that $J$ should fulfill the Kirchhoff rules,
we obtain the following bound for any given $X\in\cB$:
\begin{align}
\label{eq:sum-I}
\sum_{\substack{J\in\cJ:\\\supp J=X}}
\ee^{-\beta\cdrei\|J\|_1}
\le
\sum_{\substack{J\in(\Gamma\setminus \{0\})^X}}
\ee^{-\beta\cdrei\|J\|_1}
=
\csieben(\beta)^{|X|}
\end{align}
with the abbreviation
\begin{align}
\csieben(\beta):=\sum_{\substack{\iota\in\Gamma\setminus \{0\}}}
\ee^{-\beta\cdrei|\iota|}.
\end{align}
Because $\Gamma$ is a three-dimensional lattice, for any $k\in\N$ there are at most $\czweiundzwanzig k^2$ lattice points within distance $[\eta k,\eta {(k+1)})$ from $0$, where $\czweiundzwanzig>0$ is a constant only depending on $\Gamma$. Thus 
\begin{align}\label{eq:csieben-bd}
\csieben(\beta)
\le \sum_{k=1}^\infty \czweiundzwanzig k^2 \ee^{-\beta\cdrei \eta k }
\le \ee^{-\beta \cdreiundzwanzig}
\end{align}
for all large $\beta$ and a positive constant $\cdreiundzwanzig=\cdreiundzwanzig(\eta,\cdrei,\czweiundzwanzig)$. 
Substituting \eqref{eq:sum-I} and \eqref{eq:csieben-bd} into \eqref{eq:remaining-claim}, we obtain
\begin{align}
\sum_{\substack{J\in\cJ:\\o\in\supp J}}
\ee^{-\beta\cdrei\size J}
\le \sum_{\substack{X\in\cB:\\o\in X}}\ee^{-\beta \cdreiundzwanzig |X|}.
\end{align}
The last sum is estimated with the following Peierls argument:
Let $M<\infty$ be the maximal vertex degree in the mesoscopic lattice with 
edge set $\Lambda$.
Let $n\in\N$.
For every set $X\in \cB$ with $o\in X$ and $|X|=n$, 
there is a closed path of length $2n$ steps that starts in $o$ and
visits every edge in $X$.
There are at most $M^{2n}$ choices of closed paths of length $2n$ 
starting in $o$, and therefore at most $M^{2n}$ choices of $X$.
We conclude for all large $\beta$:
\begin{align}
\operatorname{lhs} \eqref{eq:claim-Peierls}
&\le\sup_{E\Subset\Lambda}\sup_{o\in E}
        \sum_{\substack{X\in\cB:\\o\in X}}\ee^{-\beta \cdreiundzwanzig |X|}
  \le 
\sum_{n=1}^\infty M^{2n}\ee^{-\beta \cdreiundzwanzig n}
=\frac{M^2\ee^{-\beta \cdreiundzwanzig}}{1-M^2\ee^{-\beta \cdreiundzwanzig}}\le 2M^2\ee^{-\beta \cdreiundzwanzig}
\le \ee^{-\beta \cvierundzwanzig}
\end{align}
with $\cvierundzwanzig>0$ only depending on $\cdreiundzwanzig $ and $M$. This proves the claim \eqref{eq:claim-Peierls}. 

Next, we prove claim \eqref{eq:hilfsbeh-sum-K-exp-a}. We observe that 
\eqref{eq:supp-vs-1-norm} and \eqref{eq:eq-1-norm-size} imply 
$\eta|\supp J|\le\|J\|_1\le\size J$. Using this, $\ceins-\cfuenf=\cdrei$, and 
\eqref{eq:claim-Peierls}, the claim \eqref{eq:hilfsbeh-sum-K-exp-a} follows 
from the estimate
\begin{align}
&\sum_{\substack{J\in \cJ_m:\\o\in\supp J}} |K(J)|\ee^{a(J)}
=\sum_{\substack{J\in\cJ_m:\\o\in\supp J}}
\ee^{-\beta(\ceins\size J-\cfuenf \eta |\supp J|)}\nonumber\\
\le &\sum_{\substack{J\in\cJ:\\o\in\supp J}}
\ee^{-\beta(\ceins-\cfuenf)\size J}
=\sum_{\substack{J\in\cJ:\\o\in\supp J}}
\ee^{-\beta\cdrei\size J}\le \ee^{-\beta \cvierundzwanzig},
\qquad (m\in\N).
\label{eq:split-sum}
\end{align}
Note that we have dropped the index $m$ in the last two sums. 

To verify \eqref{eq:assumption-cluster} for $\tilde\cJ=\cJ_m$,  
we define the closure of any edge set $F\subseteq E$ by 
\begin{align}
\overline{F}:=\{f\in E|f \text{ has a common vertex with some }e\in F\}. 
\end{align}
Let $m\in\N$ and $I\in\cJ_m$. Summing \eqref{eq:split-sum} over 
$o\in\overline{\supp I}$, we conclude
\begin{align}
&\sum_{J\in \cJ_m} |K(J)|1_{\{I\sim J\}}\ee^{a(J)}
\le \sum_{o\in\overline{\supp I}}
\sum_{\substack{J\in \cJ_m:\\o\in\supp J}} |K(J)|\ee^{a(J)}\nonumber\\
\le& \ee^{-\beta \cvierundzwanzig}|\overline{\supp I}|\le \ee^{-\beta \cvierundzwanzig}M|\supp I|
\le
\beta \cfuenf \eta |\supp I|=a(I)
\end{align}
for all large $\beta$, uniformly in $I\in\cJ_m$. Here we have used 
that $\ee^{-\beta \cvierundzwanzig}M\le \beta \cfuenf \eta$ for large $\beta$. 
\end{proof}

\begin{lemma}[Exponential decay of partial partition sums]
\label{le:exp-convergence-zplus}
For all sufficiently large $\beta>0$ and $m\in\N\cup\{\infty\}$, the following holds with the constants $\ceins=c/(2\czwei)$ and $\cvierundzwanzig$ as in Lemma
\ref{le:hypothesis-convergence-cluster}:
\begin{align}
\label{eq:claim-bound-z-beta+}
\sup_{m\in\N\cup\{\infty\}}\sup_{E\Subset\Lambda}\sup_{o\in E}
\sum_{\substack{I\in\cI:\\o\in\supp I}}\ee^{\beta\ceins\size I}z^+_m(\beta,I)
\le \ee^{-\beta \cvierundzwanzig}.
\end{align}
In particular, in this case, $z_m(\beta,I)$ is well-defined
for all $I\in\cI$ and fulfills the same bound 
\begin{align}
\label{eq:claim-bound-z-beta}
\sup_{m\in\N\cup\{\infty\}}\sup_{E\Subset\Lambda}\sup_{o\in E}
\sum_{\substack{I\in\cI:\\o\in\supp I}}\ee^{\beta\ceins\size I}|z_m(\beta,I)|
\le \ee^{-\beta \cvierundzwanzig}.
\end{align}
\end{lemma}
\begin{proof}
Using \eqref{eq:monotone-conv-z+m} and monotone convergence for series, it 
suffices to consider only finite 
$m\in\N$ to prove \eqref{eq:claim-bound-z-beta+}.
Let $\beta>0$, $I\in\cI\setminus\{0\}$, and $m\in\N$.
Inserting \eqref{eq:estimate-Hdisl-size} into the definition 
\eqref{eq:def-zmplus} of $z_m^+$  
yields 
\begin{align}
\label{eq:first-est-z-beta}
z^+_m(\beta,I)\le\sum_{n=1}^\infty
\sum_{\substack{(I_1,\ldots,I_n)\in\cJ_m^n:\\ I_1+\cdots +I_n=I}}
|U(I_1,\ldots,I_n)|
\prod_{j=1}^n\ee^{-\beta \frac{c}{ \czwei} \size I_j}.
\end{align}
For $(I_1,\ldots,I_n)\in\cJ^n_m$ 
with $U(I_1,\ldots,I_n)\neq 0$ 
and $I=I_1+\cdots +I_n$ as in the above summation, 
we have 
\begin{align}
\sum_{j=1}^n \size I_j\ge \size I
\end{align}
from \eqref{eq:size-subadditive}, and hence, taking again the constant 
$\ceins=c/(2\czwei)$
\begin{align}
\label{eq:decompose-product}
\prod_{j=1}^n\ee^{-\beta \frac{c}{ \czwei} \size I_j}
\le \ee^{-\beta\ceins\size I} \prod_{j=1}^n\ee^{-\beta\ceins\size I_j}.
\end{align}
We choose a reference edge $o\in\supp I$.
Substituting \eqref{eq:decompose-product} and \eqref{eq:def-K} in 
\eqref{eq:first-est-z-beta} yields 
\begin{align}
\label{eq:second-est-z-beta}
z^+_m(\beta,I)\le & \ee^{-\beta\ceins\size I} 
\sum_{n=1}^\infty\sum_{\substack{(I_1,\ldots,I_n)\in\cJ_m^n:\\ I_1+\cdots +I_n=I}}
|U(I_1,\ldots,I_n)|\prod_{j=1}^nK(I_j)
.
\end{align}
The inner sum on the right-hand side can be extended to run over all 
$n$-tuples $(I_1,\ldots,I_n)$ in 
$\cJ_m^n$ with $o\in\supp I_1\cup\ldots\cup\supp I_n$, since by definition, 
for any $I$ which cannot be written as a sum 
of such $I_1,\ldots, I_n$, for some $n\in\N$, $z^+_m(\beta,I)=0$ or $o\notin\supp I$. 
It follows
\begin{align}
\sum_{\substack{I\in\cI:\\o\in\supp I}}\ee^{\beta\ceins\size I}z^+_m(\beta,I)
\le 
\sum_{n=1}^\infty\sum_{\substack{(I_1,\ldots,I_n)\in\cJ_m^n:\\
o\in\supp I_1\cup\ldots\cup\supp I_n}}
|U(I_1,\ldots,I_n)|\prod_{j=1}^nK(I_j)
=:C_{m,o,\beta}.
\end{align}

As we observed above, it suffices to consider only finite $m$ in the claim 
\eqref{eq:claim-bound-z-beta+}. 
It remains to show that for $\beta$ large enough it is true that
\begin{align}
\label{eq:reduced-claim1}
\sup_{m\in\N}\sup_{E\Subset\Lambda}
\sup_{o\in E}
C_{m,o,\beta}
\le \ee^{-\beta \cvierundzwanzig}.
\end{align}
Note that this condition does not involve $I$.
Because $|U(I_1,\ldots,I_n)|$ is invariant under permutation of its
arguments, we can bound $C_{m,o,\beta}$ by
\begin{align}
C_{m,o,\beta}\le&
\sum_{n=1}^\infty 
n \sum_{\substack{(I_1,\ldots,I_n)\in\cJ_m^n:\\o\in\supp I_1}}
|U(I_1,\ldots,I_n)|\prod_{j=1}^nK(I_j)
\nonumber\\
=&
\sum_{\substack{I_1\in\cJ_m:\\o\in\supp I_1}} K(I_1) \bigg(1+\sum_{n=2}^\infty 
n \sum_{(I_2,\ldots,I_n)\in\cJ_m^{n-1}}
|U(I_1,\ldots,I_n)|\prod_{j=2}^nK(I_j)\bigg);
\label{eq:auxiliary-partition-sum}
\end{align}
for the summand indexed by $n=1$ we have used $U(I_1)=1$.

By Lemma~\ref{le:hypothesis-convergence-cluster}, we may apply 
Fact~\ref{fact:velenik-conv-clusterexpansion}, yielding
\begin{align}
\label{eq:cluster-convergence}
1+\sum_{n=2}^\infty 
n \sum_{(I_2,\ldots,I_n)\in\cJ_m^{n-1}}
|U(I_1,\ldots,I_n)|\prod_{j=2}^nK(I_j)\le \ee^{a(I_1)}.
\end{align}
Combining \eqref{eq:auxiliary-partition-sum}, \eqref{eq:cluster-convergence}, and \eqref{eq:hilfsbeh-sum-K-exp-a} from 
Lemma~\ref{le:hypothesis-convergence-cluster} gives
\begin{align}
\label{eq:bound-K-exp-a}
\sup_{m\in\N}\sup_{E\Subset\Lambda}\sup_{o\in E}C_{m,o,\beta}
\le \sup_{m\in\N}\sup_{E\Subset\Lambda}\sup_{o\in E} 
\sum_{\substack{I_1\in\cJ_m:\\o\in\supp I_1}}K(I_1)\ee^{a(I_1)}
\le \ee^{-\beta \cvierundzwanzig},
\end{align} 
yielding claim~\eqref{eq:reduced-claim1}.
Since $|z_m(\beta,I)|\le z_m^+(\beta,I)$, the claim \eqref{eq:claim-bound-z-beta} 
is an immediate consequence of \eqref{eq:claim-bound-z-beta+}. 
\end{proof}

\subsection{Gaussian lower bound for Fourier transforms}
Next, we apply a cluster expansion with $K(I,\phi)$ defined in \eqref{eq:KIphi}
to obtain a representation of $\cZ_{\beta,\phi}$ and finally a bound for the Fourier 
transform of the observable. 

\begin{lemma}[Partition sums in the presence of $\phi$]
\label{le:Z-beta-real}
For all $\beta$ large enough the following identity holds for any 
$\phi\in\R^E$:
\begin{align}
\label{eq:z-cos}
0<\cZ_{\beta,\phi}= &\exp\left(\sum_{I\in\cI}z(\beta,I)\ee^{\ii\scalara{\phi}{I}}\right)
= \exp\left(\sum_{I\in\cI}z(\beta,I)\cos\scalara{\phi}{I}\right)
\nonumber\\\le& \exp\left(\sum_{I\in\cI}z^+(\beta,I)\right)
<\infty.
\end{align}
\end{lemma}
\begin{proof}
Recall the definitions \eqref{eq:zm-series} and \eqref{eq:def-zmplus} of 
$z_m$ and $z_m^+$. Take any $\phi\in\R^E$.
Rearranging a multiple series with positive summands and using 
\eqref{eq:claim-bound-z-beta+} for $\beta$ large enough, we obtain 
\begin{align}
\label{eq:z+summable}
\sum_{n=1}^\infty\sum_{(I_1,\ldots,I_n)\in\cJ^n} 
|U(I_1,&\ldots,I_n)| \prod_{j=1}^n \ee^{-\beta H_{\rm disl}(I_j)}
=\sum_{I\in\cI}z^+(\beta,I)\le |E|\ee^{-\beta\cvierundzwanzig}<\infty.
\end{align}
Using this as a dominating series and the fact $|K(I,\phi)|=\ee^{-\beta H_{\rm disl}(I)}$, 
the following rearrangement of the series is valid for all $m\in\N$:
\begin{align}
\label{eq:sum-z-m-exp}
\sum_{n=1}^\infty\sum_{(I_1,\ldots,I_n)\in\cJ_m^n} 
U(I_1,\ldots,I_n) \prod_{j=1}^n K(I_j,\phi)
=\sum_{I\in\cI}z_m(\beta,I)\ee^{\ii\scalara{\phi}{I}}.
\end{align}
By \eqref{eq:estimate-Hdisl-size}, 
\begin{align}
|K(I,\phi)|=\ee^{-\beta H_{\rm disl}(I)}  
\le \ee^{-\beta\ceins\size I}=K(I). 
\end{align}
According to Lemma \ref{le:hypothesis-convergence-cluster} and 
Facts \ref{fact:velenik-clusterexpansion} and 
\ref{fact:velenik-conv-clusterexpansion}, one has for all $m\in\N$:
\begin{align}
\exp(\operatorname{lhs} \eqref{eq:sum-z-m-exp})
&=
1+\sum_{n=1}^\infty\frac{1}{n!}\sum_{(I_1,\ldots,I_n)\in(\cJ_m)^n_{\not\sim}}
\prod_{j=1}^n K(I_j,\phi).
\label{eq:cluster-exp2}
\end{align}
From monotone convergence, we know 
\begin{align}
1+\sum_{n=1}^\infty\frac{1}{n!}\sum_{(I_1,\ldots,I_n)\in(\cJ_m)^n_{\not\sim}}
\prod_{j=1}^n |K(I_j,\phi)|
\stackrel{m\to\infty}{\longrightarrow}
1+\sum_{n=1}^\infty\frac{1}{n!}\sum_{(I_1,\ldots,I_n)\in\cJ^n_{\not\sim}}
\prod_{j=1}^n |K(I_j,\phi)| <\infty;
\label{eq:partial-sum-m-to-infty}
\end{align}
the finiteness follows as in the argument below \eqref{eq:Zsigmaxy}.  
Consequently, applying dominated convergence in \eqref{eq:cluster-exp2}
and using $\cJ_m\uparrow \cJ$ and \eqref{eq:Zsigmaxy} yields
\begin{align}
\exp(\operatorname{lhs} \eqref{eq:sum-z-m-exp})
\stackrel{m\to\infty}{\longrightarrow}&
1+\sum_{n=1}^\infty\frac{1}{n!}\sum_{(I_1,\ldots,I_n)\in\cJ^n_{\not\sim}}
\prod_{j=1}^n K(I_j,\phi)
=\cZ_{\beta,\phi}.
\end{align}
On the other hand, 
from \eqref{eq:z+summable} and dominated convergence for series,
\begin{align}
\label{eq:zm-to-z-series}
\sum_{I\in\cI}z_m(\beta,I)\ee^{\ii\scalara{\phi}{I}}
\stackrel{m\to\infty}{\longrightarrow}
\sum_{I\in\cI}z(\beta,I)\ee^{\ii\scalara{\phi}{I}}.
\end{align}
Taking the limit $m\to\infty$ in equation \eqref{eq:sum-z-m-exp} 
yields the first equality in claim \eqref{eq:z-cos}.  

The second equality of claim \eqref{eq:z-cos} follows from the following 
symmetry consideration. One has 
$I\in\cI$ if and only if $-I\in\cI$ and $z(\beta,I)=z(\beta,-I)$ by definition, hence
\begin{align}
\sum_{I\in\cI}z(\beta,I)\ee^{\ii\scalara{\phi}{I}}=&\frac12\left[\sum_{I\in\cI}z(\beta,I)\ee^{\ii\scalara{\phi}{I}}+\sum_{I\in\cI}z(\beta,-I)\ee^{\ii\scalara{\phi}{-I}}\right]=\sum_{I\in\cI}z(\beta,I)\cos\scalara{\phi}{I}.
\end{align}
The last series converges absolutely and its absolute value is bounded by 
$\sum_{I\in\cI}z^+(\beta,I)<\infty$. 
\end{proof}

\begin{lemma}[Gaussian lower bound for Fourier transforms]
\label{le:fourier-observable}
For all $\beta$ large enough, the following holds for any $\sigma\in\R^E$:
\begin{align}
\label{eq:lower-bound-fourier-trafo-sigma-I}
\EE_{P_\beta}\Big[\ee^{\ii\scalara{\sigma}{I}}\Big]\ge
\exp\left(-\frac{1}{2}\sum_{I\in\cI}|z(\beta,I)|\scalara{\sigma}{I}^2\right)
.
\end{align}
\end{lemma}
\begin{proof}
By \eqref{eq:ratio-Zsigma-Z0}, \eqref{eq:def-Z-beta-sigma},
and Lemma~\ref{le:Z-beta-real}, we have 
\begin{align}
\label{eq:Fourier-trafo-as-ratio}
\EE_{P_\beta}\Big[\ee^{\ii\scalara{\sigma}{I}}\Big]
=&\frac{Z_\beta(\sigma)}{Z_\beta(0)}=
\frac{\bbE[\cZ_{\beta,\sigma+\phi}]}{\bbE[\cZ_{\beta,\phi}]},
\\
\cZ_{\beta,\sigma+\phi}
=&\exp\left(\sum_{I\in\cI}z(\beta,I)
\cos\scalara{\sigma+\phi}{I}\right) 
.
\label{eq:Z-sigma-xy}
\end{align}
Using 
\begin{align}
\cos\scalara{\sigma+\phi}{I}
=\cos\scalara{\phi}{I}  (\cos\scalara{\sigma}{I}-1)
+\cos\scalara{\phi}{I}  
-\sin\scalara{\phi}{I} \sin\scalara{\sigma}{I}
\end{align}
and the bound 
\begin{align}
\cos\scalara{\phi}{I}  (\cos\scalara{\sigma}{I}-1)
\ge& 
-|\cos\scalara{\sigma}{I}-1|
\ge
-\frac{1}{2}\scalara{\sigma}{I}^2, 
\end{align}
we obtain
\begin{align}
&\sum_{I\in\cI}z(\beta,I)
\cos\scalara{\sigma+\phi}{I}
\nonumber\\\ge& 
-\frac{1}{2}\sum_{I\in\cI}|z(\beta,I)|\scalara{\sigma}{I}^2
+\sum_{I\in\cI}z(\beta,I)
\left[\cos\scalara{\phi}{I}  
-\sin\scalara{\phi}{I} \sin\scalara{\sigma}{I} \right]
.
\label{eq:lower-bound-addition-thm}
\end{align}
We take the average over an auxiliary sign $\Sigma$ taking values $\pm1$.
We substitute $\phi$ by $\Sigma\phi$ in \eqref{eq:lower-bound-addition-thm}.
Then, using  the facts 
$\cos\scalara{\Sigma\phi}{I}=\cos\scalara{\phi}{I}$ and 
$\sin\scalara{\Sigma\phi}{I}=\Sigma\sin\scalara{\phi}{I}$, 
it follows
\begin{align}
&\frac12\sum_{\Sigma\in\{\pm 1\}}\cZ_{\beta,\sigma+\Sigma \phi}
=\frac12\sum_{\Sigma\in\{\pm 1\}}\exp\left(\sum_{I\in\cI}z(\beta,I)
\cos\scalara{\sigma+\Sigma\phi}{I}\right)
\nonumber\\\ge& 
\exp\left(-\frac{1}{2}\sum_{I\in\cI}|z(\beta,I)|\scalara{\sigma}{I}^2\right)
\cdot\frac12\sum_{\Sigma\in\{\pm 1\}}\exp\left(\sum_{I\in\cI}z(\beta,I)
\big[\cos\scalara{\phi}{I}  
-\sin\scalara{\Sigma \phi}{I} \sin\scalara{\sigma}{I} \big]\right)
\nonumber\\
=& 
\exp\left(-\frac{1}{2}\sum_{I\in\cI}|z(\beta,I)|\scalara{\sigma}{I}^2\right)
\cZ_{\beta,\phi}
\cdot\frac12\sum_{\Sigma\in\{\pm 1\}}
\exp\left(-\Sigma\sum_{I\in\cI}z(\beta,I)
\sin\scalara{\phi}{I} \sin\scalara{\sigma}{I}\right)
.
\end{align}
Note that $\sum_{I\in\cI}z(\beta,I)\sin\scalara{\phi}{I} \sin\scalara{\sigma}{I}$
converges absolutely for $\beta$ large enough by Lemma \ref{le:Z-beta-real}. 
Since $\frac12(\ee^x+\ee^{-x})\ge 1$ for all $x\in\R$, we obtain
\begin{align}
\frac12\sum_{\Sigma\in\{\pm 1\}}\cZ_{\beta,\sigma+\Sigma \phi}
\ge & 
\exp\left(-\frac{1}{2}\sum_{I\in\cI}|z(\beta,I)|\scalara{\sigma}{I}^2\right)
\cZ_{\beta,\phi}.
\end{align}
Using that $\phi$ is centered
Gaussian and $\cZ_{\beta,\phi}$ is bounded and positive, we conclude
\begin{align}
\bbE[\cZ_{\beta,\sigma+\phi}]
=&\bbE\left[\frac12\sum_{\Sigma\in\{\pm 1\}}\cZ_{\beta,\sigma+\Sigma \phi}\right]
\nonumber\\\ge& 
\exp\left(-\frac{1}{2}\sum_{I\in\cI}|z(\beta,I)|\scalara{\sigma}{I}^2\right)
\bbE[\cZ_{\beta,\phi}]
.
\label{eq:Z-sigma-xy2}
\end{align}
In view of 
\eqref{eq:Fourier-trafo-as-ratio} and $\bbE[\cZ_{\beta,\phi}]>0$,
 this proves the claim.
\end{proof}

\section{Proof of the main result} 
\label{sec:bounding-obs}
\subsection{Bounding the observable}
For $I\in\cI$ and $b=b(I)$ as in \eqref{eq:b-as-fn-of-I}, we take
the minimizer $w^*:\R^3\to\R^{3\times 3}$ defined in \eqref{eq:def-w-star}
and set $w^*(x,I):=w^*(x)$. 
For arbitrary $x,y\in\R^3$, we choose 
$\sigma(x,y)=(\sigma_{ij}(x,y))_{i,j\in[3]}\in(\R^E)^{[3]\times[3]}$ 
satisfying the equation
\begin{align}
\scalara{\sigma_{ij}(x,y)}{I}=w^*_{ij}(x,I)-w^*_{ij}(y,I)
\end{align}
for all $I\in\cI$. Such a choice is possible because $w^*(x,I)$ is a linear function 
of $I$. 

\begin{lemma}[Bounding the observable]
\label{le:w-star-bd}
There are a function $W:\Lambda\times\R^3\times\bigcup_{E\Subset\Lambda}
(\cI(E)\setminus\{0\})\to\R_{\ge 0}$ 
with 
\begin{align}
\label{eq:def-c-10}
\czehn:=\sup_{x\in\R^3}\sup_{E\Subset\Lambda}\sup_{I\in\cI(E)\setminus\{0\}}\sum_{o\in\Lambda}W(o,x,I)<\infty    
\end{align}
and $\beta_1>0$ such that for all $E\Subset\Lambda$, $I\in\cI(E)\setminus\{0\}$, $o\in\supp I$, $x\in\R^3$, $\beta\ge \beta_1$
and $i,j\in[3]$ we have
\begin{align}
\label{eq:bound-wij-star}
w^*_{ij}(x,I)^2\le W(o,x,I)\ee^{\beta\ceins\size I}.   
\end{align}
\end{lemma}
\begin{proof}
Take $E\Subset\Lambda$, $I\in\cI(E)\setminus\{0\}$, $o\in\supp I$, $x\in\R^3$, and $i,j\in[3]$. 
We choose a vertex $v(o)\in o$. We set 
\begin{align}
M_1(I,o):= & \max_{i,j\in[3]}\sum_{l=1}^3\int_{\R^3}|u-v(o)||b_{lij}(u)|\, du, \\
R(I,o):= & \max\{|x|:x\in\supp\varphi\}+\max\{|v'-v(o)|:v'\in e\text{ for some }
e\in\supp I\}. 
\end{align}
Since $\size I$ is bounded away from $0$, there is a constant $\cdreizehn>0$
such that $R(I,o)\le\cdreizehn\size I$. 
By the definition \eqref{eq:b-as-fn-of-I} of $b(I)$, one has the bound 
$\|b_{lij}(I)\|_1\le\cvierzehn\|I\|_1\le \cvierzehn\size I$
for all its components, with some constant $\cvierzehn>0$. Hence, we obtain 
\begin{align}
M_1(I,o)\le R(I,o)\max_{i,j\in[3]}\sum_{l=1}^3\|b_{lij}(I)\|_1
\le 3\cdreizehn\cvierzehn (\size I)^2. 
\end{align}
Because $b$ is compactly supported and divergence-free in the sense of equation 
\eqref{eq:div-tilde-b-zero},
\begin{align}
Q(I):=\int_{\R^3}b(I)(u)\, du=0
\end{align}
by the fundamental theorem of calculus. 
Recall the representation $w^*=w^b+d_0\psi^*$ with $w^b$, $d_0\psi^*$ as in 
\eqref{eq:components-of-w-b-componentwise} and \eqref{eq:components-of-d0-psi-stern-componentwise}.
We now establish the bound~\eqref{eq:bound-wij-star} in two steps,
first for $x$ far from $o$, then close to~$o$.
A key estimate is provided by bounds on integral kernels
proven in Appendix~\ref{sec:integral-kernels}.

\textit{Case 1:} First we consider the case $|v(o)-x|\ge 2R(I,o)$. 
It follows from \eqref{eq:components-of-w-b-componentwise} and the second inequality 
in \eqref{eq:bound-partial-r1} from Lemma~\ref{le:simplified-dipole-expansion} 
(see Appendix \ref{sec:integral-kernels})
in the case that $v(o)=0$ is the origin that
\begin{align}
\label{eq:intermediate-bound-wb}
|w^b_{ij}(x)|\le\sum_{l=1}^3|\partial_l\Delta^{-1}b_{lij}(I)(x)|
\le \frac{24M_1(I,o)}{\pi}\frac{1}{|v(o)-x|^3}
\le\frac{\cfuenfzehn(\size I)^2}{|v(o)-x|^3}
\end{align}
with the constant $\cfuenfzehn=72\cdreizehn\cvierzehn/\pi$. 

The stability condition \eqref{eq:def-F} implies $|\lambda|/|2\mu+\lambda|\le 1$.
In the same way as in \eqref{eq:intermediate-bound-wb}, \eqref{eq:components-of-d0-psi-stern-componentwise} and 
the second 
inequality in \eqref{eq:bound-partial-r1-strich} from Lemma 
\ref{le:dipole-expansion-wzwei} (see again Appendix \ref{sec:integral-kernels})
in the case that $v(o)=0$ is the origin give
\begin{align}
|d_0\psi^*_{ij}(x)|&\le\sum_{k,l=1}^3\bigg[|\partial_i\Delta^{-2}\partial_k\partial_lb_{ljk}(I)(x)|
+|\partial_i\Delta^{-2}\partial_j\partial_lb_{lkk}(I)(x)|\bigg]\nonumber\\
&\le \frac{2\cdot 9\cdot 36 M_1(I,o)}{\pi}\frac{1}{|v(o)-x|^3}
\le\frac{\csechzehn(\size I)^2}{|v(o)-x|^3}
\end{align}
with the constant $\csechzehn=1944\cdreizehn\cvierzehn/\pi$. 
It follows still in the case $v(o)=0$
\begin{align}
\label{eq:estimates-w1-w2}
|w^*_{ij}(x,I)|\le |w^b_{ij}(x)|+|d_0\psi^*_{ij}(x)|
\le \frac{(\cfuenfzehn+\csechzehn)(\size I)^2}{|v(o)-x|^3}. 
\end{align}
The next step involves translation-invariance: Shifting both $x$ and $I$ by 
a mesoscopic lattice vector $v\in V_\Lambda$ does not change $w^*_{ij}(x,I)$
because $(x,I)\mapsto b(I)(x)$ has the same translation-invariance. 
Because the inequality \eqref{eq:estimates-w1-w2} is written in a 
translation-invariant form, it holds 
also if we drop the assumption $v(o)=0$. This yields
\begin{align}
w^*_{ij}(x,I)^2\le\frac{(\cfuenfzehn+\csechzehn)^2(\size I)^4}{|v(o)-x|^6}
\le |v(o)-x|^{-6}\ee^{\beta\ceins\size I}
\end{align}
for all $\beta\ge\beta_1$ for sufficiently large $\beta_1$, neither depending 
on $o$, $x$, nor $I$.

\textit{Case 2:} Next we consider the case $|v(o)-x|<2R(I,o)$. 
We recall the definition of $J_{jk}(I)$ from \eqref{eq:def-J-jk-I}. 
We now use the symbol $\|{\cdot}\|_1$ in two different ways. On the one hand, $\|I\|_1=\sum_{e\in E}|I_e|$ for $I$. On the other hand, $\|J_{jk}(I)\|_1$ denotes the total unsigned
mass of the signed measure $J_{jk}(I)$ given by the following definition:
For any signed measure $\tilde J$ on $\R^3$ with Hahn decomposition 
$\tilde J=\tilde J_+-\tilde J_-$, we define
$\|\tilde J\|_1:=\tilde J_+(\R^3)+\tilde J_-(\R^3)$. 
With this interpretation, we have 
\begin{align}
\|J_{jk}(I)\|_1\le\sup_{e\in\Lambda}\lambda_e(\R^3)\|I\|_1. 
\label{eq:total-mass-J}
\end{align}
Combining this with \eqref{eq:bound-partial-delta-phiJ} and 
\eqref{eq:bound-partial-dreifach-delta-phiJ} from Lemma \ref{le:uniform-bounds} 
in Appendix \ref{sec:appendix-kernels-diag} yields the bound 
\begin{align}
\label{eq:bound-w-star-diag}
|w^*_{ij}(x,I)|\le\cneunzehn\|I\|_1\le\cneunzehn\size I
\end{align}
for all $x\in\R^3$ and all $I\in\cI\setminus\{0\}$ with a constant
$\cneunzehn>0$. Note that 
\begin{align}
\label{eq:bound-on-number-of-points}
|\{o\in\Lambda:|v(o)-x|<2R(I,o)\}|
\le\czwanzig(\size I)^3
\end{align}
with a constant $\czwanzig>0$ depending on the lattice spacing in $\Lambda$. 
Squaring \eqref{eq:bound-w-star-diag}, we obtain 
\begin{align}
|w^*_{ij}(x,I)|^2\le(\size I)^{-3}\cneunzehn^2(\size I)^5\le (\size I)^{-3}
\ee^{\beta\ceins\size I}
\end{align}
again for all $\beta\ge\beta_1$ for sufficiently large $\beta_1$, neither 
depending on $o$, $x$, nor $I$.

Combining the two cases, the claim \eqref{eq:bound-wij-star} holds for 
\begin{align}
W(o,x,I):=1_{\{|v(o)-x|\ge 2R(I,o)\}}|v(o)-x|^{-6}
+1_{\{|v(o)-x|<2R(I,o)\}}(\size I)^{-3}.    
\end{align}
To bound $\sum_{o:|v(o)-x|\ge 2R(I,o)}|v(o)-x|^{-6}$, observe that $R(I,o)$ is 
bounded away from 0 uniformly in $I$ and $o$, and that $|v(o)-x|^{-6}$ is 
summable away from $x$ in three dimensions. 
We use \eqref{eq:bound-on-number-of-points} to bound 
$\sum_{o:|v(o)-x|<2R(I,o)}(\size I)^{-3}$.  
We conclude
\begin{align}
\sum_{o\in\Lambda}W(o,x,I)\le\czehn
\end{align}
uniformly in $x$, $E\Subset\Lambda$, and $I\in\cI(E)$,
with a constant $\czehn$ depending on $\Lambda$. 
\end{proof}

\subsection{Identifying long-range order}
We finally prove now our main result. 
\begin{proof}[Proof of Theorem~\ref{thm:main-result}]
Applying Lemma~\ref{le:fourier-observable} yields for $t\in\R$, 
$E\Subset\Lambda$, $x,y\in\R^3$, and $i,j\in[3]$
\begin{align}
\EE_{P_{\beta,E}}\Big[\ee^{\ii t\scalara{\sigma_{ij}(x,y)}{I}}\Big]\ge
\exp\left(-\frac{t^2}{2}\sum_{I\in\cI(E)}|z(\beta,I)|\scalara{\sigma_{ij}(x,y)}{I}^2\right)
.
\end{align}
We may drop the summand indexed by $I=0$ because $\scalara{\sigma_{ij}(x,y)}{0}=0$. 
Inserting \eqref{eq:bound-wij-star} and employing Lemma 
\ref{le:exp-convergence-zplus} in the last line in \eqref{eq:hilfsschritt} below, we obtain for sufficiently large $\beta$ that
\begin{align}
&\frac12\sum_{I\in\cI(E)\setminus\{0\}}|z(\beta,I)|\scalara{\sigma_{ij}(x,y)}{I}^2
\le \sum_{I\in\cI(E)\setminus\{0\}}|z(\beta,I)|(w^*_{ij}(x,I)^2+w^*_{ij}(y,I)^2)   
\nonumber\\
\le& \sum_{o\in E}\sum_{\substack{I\in\cI(E):\\o\in\supp I}}|z(\beta,I)|(w^*_{ij}(x,I)^2+w^*_{ij}(y,I)^2) \nonumber\\
\le&\sup_{\tilde I\in\cI(E)\setminus\{0\}}\sum_{o\in E}(W(o,x,\tilde I)+W(o,y,\tilde I))
\sum_{\substack{I\in\cI(E):\\o\in\supp I}}\ee^{\beta\ceins\size I}|z(\beta,I)|
\nonumber\\
\le& 2\czehn
\sup_{E\Subset\Lambda}
\sup_{o\in E}\sum_{\substack{I\in\cI(E):\\o\in\supp I}}\ee^{\beta\ceins\size I}|z(\beta,I)|
\le 2\czehn \, \ee^{-\beta \cvierundzwanzig}\le 
\ee^{-\beta \celf},
\label{eq:hilfsschritt}
\end{align}
with a constant $\celf=\celf(\czehn,\cvierundzwanzig)>0$, where $\czehn$ was defined
in \eqref{eq:def-c-10}. Mind that $\celf$ does not depend on $x,y,i,j,E,\beta$.
This proves the first claim. 

By Theorem 3.3.9 in \cite{durrett}, for $\beta$ large enough, 
the variance of $\scalara{\sigma_{ij}(x,y)}{I}$ exists and fulfills 
\begin{align}
\label{eq:bound-var-proof}
& \var_{P_{\beta,E}}(\scalara{\sigma_{ij}(x,y)}{I})  
\le \EE_{P_{\beta,E}}[\scalara{\sigma_{ij}(x,y)}{I}^2]\nonumber\\
\le &-\limsup_{t\downarrow 0}t^{-2}\left(\EE_{P_{\beta,E}}[\ee^{\ii t\scalara{\sigma_{ij}(x,y)}{I}}]-2+
\EE_{P_{\beta,E}}[\ee^{-\ii t\scalara{\sigma_{ij}(x,y)}{I}}]\right)
\le \ee^{-\beta \celf},
\end{align}
where the last inequality is a consequence of 
the lower bound \eqref{eq:lower-bound-fourier-trafo-main}. 
\end{proof}

We remark that the two reflection symmetries $H_{\rm el}(-w)=H_{\rm el}(w)$ and 
$H_{\rm disl}(-I)=H_{\rm disl}(I)$ imply that $w^*(x,-I)=-w^*(x,I)$ and $w^*(x,I)$ 
are equal in distribution with respect to $P_{\beta,E}$, jointly in $x\in\R^3$. 
In particular, the first inequality in \eqref{eq:bound-var-proof}
is actually an equality.

\begin{appendix}
\section{Appendix}
\subsection{Elasticity theory}
\label{app:elasticity}

\subsubsection{Description of an elastically deformed solid in continuum approximation}
\label{sec:elastic2}

From Section~\ref{sec:elastic}, we recall that
the elastic deformation energy $E_{\mathrm{el}}(f)$ 
is modeled to be an integral over a smooth elastic energy density 
$\rho_{\mathrm{el}}:\R^{3\times 3}\to\R$: 
\begin{align}
 E_{\mathrm{el}}(f)=\int_{\R^3}\rho_{\mathrm{el}}(\nabla f(x)) \, dx.
\end{align}
We assume furthermore:
\begin{itemize}
\item $\rho_{\mathrm{el}}$ takes its minimum value $0$ at the identity matrix $\Id$.
This means that the non-deformed solid has minimal energy. 
\item The elastic energy density is insensitive to 
rotations of the solid:
\begin{align}
\label{eq:assumption-rotational-inv}\tag{rot inv}
  \rho_{\mathrm{el}}(RM)=\rho_{\mathrm{el}}(M), \quad (R\in\operatorname{SO}(3), 
M\in\operatorname{GL}_+(3)),
\end{align}
where $\operatorname{GL}_+(3)=\{M\in\R^{3\times 3}:\det M>0\}$.
In other words, deforming the solid and then rotating it costs the same 
elastic energy as only deforming it with the same deformation. 
\end{itemize}
Note that reflections $R\in\operatorname{O}(3)\setminus\operatorname{SO}(3)$ or 
singular or orientation reversing linearized deformations $M\in\R^{3\times 3}$ with 
$\det M\le 0$ do not make sense physically in this context.

For $M,N\in \operatorname{GL}_+(3)$ it is equivalent that $M^tM=N^tN$ 
and that there exists $R\in\operatorname{SO}(3)$ such that $N=RM$.
As a consequence, $\rho_{\mathrm{el}}(M)$ is 
a function of $M^tM$. We set $\tilde\rho_{\mathrm{el}}(M^tM):=\rho_{\mathrm{el}}(M)$. 

Note that the assumption \eqref{eq:assumption-rotational-inv} of rotational invariance
does \textit{not} imply isotropy of the solid, which is defined by 
\begin{align}
\label{eq:assumption-isotropy}\tag{isotropy}
  \rho_{\mathrm{el}}(MR)=\rho_{\mathrm{el}}(M), \quad (R\in\operatorname{SO}(3), 
M\in\operatorname{GL}_+(3)). 
\end{align}
Thus, anisotropy means that first rotating the solid and then deforming it with
a given deformation might cost a different elastic energy than deforming it with 
the same deformation without rotating it first. 
Although the isotropy assumption is an oversimplification for any real monocrystal,
we assume it to keep the presentation simple.

\subsubsection{Linearization}
\label{sec:linearization}

We now consider only small perturbations $f=\id+\epsilon u:\R^3\to\R^3$ 
of the identity map as deformation maps.
For $\nabla f=\Id+\epsilon \nabla u$, we obtain
\begin{align}
(\nabla f)^t\nabla f=\Id+\epsilon (\nabla u+(\nabla u)^t)
+\epsilon^2(\nabla u)^t\nabla u.
\end{align}
Substituting this into \eqref{e:F-first} gives 
\begin{align}
\tilde\rho_{\mathrm{el}}((\nabla f)^t\nabla f)
=
\epsilon^2 F(\nabla u+ (\nabla u)^t)+O(\epsilon^3),\quad (\epsilon\to 0).
\end{align}
If $\R\ni\epsilon\mapsto R_\epsilon\in\operatorname{SO}(3)$ is a path of 
rotations with $R_0=\Id$, then $\frac{d}{d\epsilon}R_\epsilon|_{\epsilon=0}$ 
is antisymmetric. Hence, the Taylor-approximated energy density 
$\epsilon^2 F(\nabla u+ (\nabla u)^t)$ is not influenced by linearized rotations. 

If the assumption \eqref{eq:assumption-isotropy} holds, then 
\begin{align}
 \tilde\rho_{\mathrm{el}}(A)=\tilde\rho_{\mathrm{el}}(R^tAR)
\end{align}
holds for all positive definite matrices $A=A^t$ and $R\in\operatorname{SO}(3)$. 
Taylor-expanded this means 
\begin{align}
F(U)=F(R^tUR) 
\end{align}
for all symmetric matrices $U=U^t$ and $R\in\operatorname{SO}(3)$. 
Thus, $F(U)$ depends only on the list $a,b,c$ of eigenvalues of $U$ 
(with multiplicities). 
For diagonal matrices $U=\diag(a,b,c)$ the only quadratic forms which are 
symmetric in $a,b,c$ are linear combinations of $(\tr U)^2=(a+b+c)^2$ and 
$|U|^2=a^2+b^2+c^2$. Thus, under an isotropy assumption, we have
\begin{align}
\label{eq:def-F-bis}
  F(U)=\frac{\lambda}{2}(\tr U)^2+\mu|U|^2 
=(a,b,c)\left(\mu\Id+\frac{\lambda}{2}ee^t\right)
  \begin{pmatrix}
 a\\ b\\ c
  \end{pmatrix}
\end{align}
with real constants $\lambda$ and $\mu$ and $e^t=(1,1,1)$. 
The matrix $\mu\Id+\frac{\lambda}{2}ee^t$ has the double eigenvalue
$\mu$ with eigenspace $e^\perp$ and a single eigenvalue
$\mu+3\lambda/2$ with eigenspace $\R e$. Hence the quadratic form $F$ is positive definite if and only if 
$\mu$ and $\lambda$ satisfy the conditions given in \eqref{eq:def-F}.
Summarizing, we have the model for the linearized elastic deformation energy
  given in \eqref{eq:def_E_el}--\eqref{eq:def_H_el}.

\subsection{From Kirchhoff's node rule to continuum sourceless currents}
\label{sec:appendix-kirchhoff}
We show that Kirchhoff's node rule for the discrete current $I$ implies absence of sources for its smoothed variant $\tilde b(I)$. 
For $e\in E$ from $x\in V$ to $y\in V$ in its counting direction,
we write $x=v_-(e)$ and $y=v_+(e)$.
We rewrite \eqref{eq:def-b-I} as 
\begin{align}
&\tilde{b}_{jk}(I)(x)=\int_{\R^3}\varphi(x-y)J_{jk}(I)(dy)
\nonumber\\
=&\sum_{e\in E} (n_e)_j(I_e)_k |v_+(e)-v_-(e)|
\int_0^1\varphi\big(x-v_-(e)-t(v_+(e)-v_-(e))\big)\,dt
\nonumber\\
=&\sum_{e\in E} (I_e)_k (v_+(e)-v_-(e))_j
\int_0^1\varphi\big(x-v_-(e)-t(v_+(e)-v_-(e))\big)\,dt,
\\&(j,k\in[3],\; I\in(\R^3)^E \text{ with }\eqref{eq:Kirchhoff},
\; x\in\R^3).\nonumber
\end{align}
As a consequence of Kirchhoff's rule \eqref{eq:Kirchhoff},
$\tilde b(I)$ is indeed divergence-free:
\begin{align}
&(\operatorname{div}\tilde{b}(I))_k(x)
\nonumber\\=&\sum_{j=1}^3\sum_{e\in E} (I_e)_k (v_+(e)-v_-(e))_j
\int_0^1\partial_j\varphi\big(x-v_-(e)-t(v_+(e)-v_-(e))\big)\,dt
\nonumber\\
=&
-\sum_{e\in E} (I_e)_k 
\int_0^1\frac{d}{dt}\varphi\big(x-v_-(e)-t(v_+(e)-v_-(e))\big)\,dt
\nonumber\\
=&
-\sum_{e\in E} (I_e)_k
\big[\varphi(x-v_+(e))-
\varphi(x-v_-(e))\big]
\nonumber\\
=&
-\sum_{e\in E} \sum_{v\in V}s_{ve}(I_e)_k\varphi(x-v)
=0.
\label{eq:div-tilde-b-zero-0}
\end{align}

\subsection{Integral kernels}
\label{sec:integral-kernels}
\subsubsection{Bounds for the dipole expansion}
\label{sec:bounds-dipole-expansion}
In this appendix, we derive a simplified version of the dipole expansion for the Coulomb potential, which we need as ingredient to identify long-distance bounds for the observable $w^*$ in the proof of Lemma~\ref{le:w-star-bd}. 

Let $\rho\in C^\infty_c(\R^3,\R)$. We define its total charge
\begin{align}
Q:=\int_{\R^3}\rho(x)\,dx.
\end{align}
We take a radius $R>0$ such that 
\begin{align}
R\ge\sup\{|x|:\;x\in\R^3,\;\rho(x)\neq 0\}
\end{align}
and the first unsigned moment
\begin{align}
M_1:=\int_{\R^3}|y||\rho(y)|\,dy.
\end{align}
The inverse Laplacian $-\Delta^{-1}$ is described by convolution with the
Coulomb potential
\begin{align}
G(y)=\frac{1}{4\pi|y|}
.
\end{align}

\begin{lemma}[Simplified dipole expansion]
\label{le:simplified-dipole-expansion}
For all $x\in\R^3$ with $|x|\ge 2R$ and $i\in[3]$ we know that 
\begin{align}
\label{eq:claim-dipole-simplified}
-\Delta^{-1}\rho(x)&=\frac{1}{4\pi}\frac{Q}{|x|}+r_1(x),\\
-\partial_i\Delta^{-1}\rho(x)&=-\frac{Q}{4\pi}\frac{x_i}{|x|^3}
+\partial_ir_1(x)
\end{align}
with the error bounds
\begin{align}
\label{eq:bound-partial-r1}
|r_1(x)|\le\frac{M_1}{\pi}\frac{1}{|x|^2},
\qquad
|\partial_i r_1(x)|\le\frac{8M_1}{\pi}\frac{1}{|x|^3}.
\end{align}
\end{lemma}
\begin{proof}
Let $x\in\R^3$ with $|x|\ge 2R$ and $y\in\operatorname{supp}\rho$.
Then for all $t\in[0,1]$, we have
\begin{align}
|x-ty|\ge |x|-R\ge \frac{|x|}{2}\ge R.
\end{align}
We apply 
\begin{align}
\label{eq:taylor1}
f_{x,y}(1)=f_{x,y}(0)+\int_0^1 f'_{x,y}(t)\,dt
\end{align}
to
\begin{align}
f_{x,y}(t):=\frac{1}{|x-ty|}.
\end{align}
In this proof, $\partial_i=\partial/\partial x_i$ always refers to 
the variable $x_i$.
Using $\partial_i|x|^\alpha=\alpha x_i|x|^{\alpha-2}$,
we calculate:
\begin{align}
\label{eq:fxy'}
f_{x,y}'(t)&=\frac{\scalara{x-ty}{y}}{|x-ty|^3},
\qquad\partial_i f_{x,y}(t)=-\frac{x_i-ty_i}{|x-ty|^3},
\\
\label{eq:partial-i-fxy'}
\partial_i f_{x,y}'(t)
&=
\frac{y_i}{|x-ty|^3}
-3\frac{\scalara{x-ty}{y}(x_i-ty_i)}{|x-ty|^5}.
\end{align}
Consequently, we can bound the integrand
in \eqref{eq:taylor1} and its $\partial_i$-derivative as follows: 
\begin{align}
|f_{x,y}'(t)|&=
\frac{|\scalara{x-ty}{y}|}{|x-ty|^3}\le \frac{|y|}{|x-ty|^2}\le
\frac{4|y|}{|x|^2},
\\
|\partial_i f_{x,y}'(t)|&\le \frac{4|y|}{|x-ty|^3}\le \frac{32|y|}{|x|^3}.
\label{eq:partial-i-f-xy-strich}
\end{align}
We obtain
\begin{align}
\label{eq:taylor-integrated}
-\Delta^{-1}\rho(x)
=
\frac{1}{4\pi}\int_{\R^3} \frac{1}{|x-y|}\rho(y)\,dy=
\frac{1}{4\pi}\int_{\R^3} \left(f_{x,y}(0)+\int_0^1 f'_{x,y}(t)\,dt\right)
\rho(y)\,dy
\end{align}
and hence the claim~\eqref{eq:claim-dipole-simplified} holds with the error term
\begin{align}
r_1(x)=\frac{1}{4\pi}\int_{\R^3}\int_0^1f_{x,y}'(t)\,dt\,
\rho(y)\,dy.
\end{align}
It is bounded as follows:
\begin{align}
|r_1(x)|\le\frac{1}{4\pi}\int_{\R^3}\int_0^1\frac{4|y|}{|x|^2}
\,dt\,|\rho(y)|\,dy
=\frac{M_1}{\pi}\frac{1}{|x|^2}.
\end{align}
This shows the first claimed bound in~\eqref{eq:bound-partial-r1}.

Substituting \eqref{eq:partial-i-f-xy-strich} in  the $\partial_i$-derivative
of equation~\eqref{eq:taylor-integrated}, i.e., in 
\begin{align}
&-\partial_i\Delta^{-1}\rho(x)
=
\frac{1}{4\pi}\int_{\R^3} \left(\partial_i f_{x,y}(0)+
\int_0^1 \partial_if'_{x,y}(t)\,dt\right)
\rho(y)\,dy,
\end{align}
the second claimed bound in~\eqref{eq:bound-partial-r1} follows, too.

\end{proof}

\subsubsection{Properties of $\partial_i\partial_j\Delta^{-2}$}
Next we describe $\partial_i\partial_j\Delta^{-2}$ explicitly by an integral kernel. 
We derive it from its Fourier transform, cf. \eqref{eq:def-Fouriertrafo}. 

\begin{lemma}[Integral kernel of $\partial_i\partial_j\Delta^{-2}$]
For any Schwartz functions $f,g:\R^3\to\C$, we have 
\begin{align}
\label{eq:claim-partial2Delta-2}
-\int_{\R^3}\overline{\hat g(k)}\frac{k_ik_j}{|k|^4}
\hat f(k)\,dk=
\frac{1}{8\pi}\int_{\R^3}\int_{\R^3}
\overline{g(x)}
\frac{1}{|x-y|}
\left(\frac{(x_i-y_i)(x_j-y_j)}{|x-y|^2}-\delta_{ij}\right)
f(y)\,dx\,dy
.
\end{align}
\end{lemma}
\begin{proof}
Using
\begin{align}
\frac{1}{|k|^4}=\int_0^\infty \ee^{-\frac12 z|k|^2}\frac{z}{4}\,dz
\end{align}
for $k\in\R^3\setminus\{0\}$,
we rewrite the left hand side in the claim~\eqref{eq:claim-partial2Delta-2}
as follows:
\begin{align}
\operatorname{lhs} \eqref{eq:claim-partial2Delta-2}=&
-\int_{\R^3}\int_0^\infty \overline{\hat g(k)}
k_ik_j \ee^{-\frac12 z|k|^2}\frac{z}{4}
\hat f(k)\,dz\,dk
\nonumber\\=&
-\int_0^\infty\frac{z}{4}\int_{\R^3} \overline{\hat g(k)}
k_ik_j \ee^{-\frac12 z|k|^2}
\hat f(k)\,dk\,dz
\qquad\text{(by Fubini)}
\nonumber\\=&
\frac{1}{4(2\pi)^{\frac32}}
\int_0^\infty z^{-\frac12} \int_{\R^3}
\overline{\hat g(k)}
\left[\partial_i\partial_j
\left(\ee^{-\frac{|\cdot|^2}{2z}}\ast f\right)\right]^{\wedge}(k)
\,dk\,dz
\nonumber\\=&
\frac{1}{4(2\pi)^{\frac32}}
\int_0^\infty z^{-\frac12} \int_{\R^3}
\overline{g(x)}
\partial_i\partial_j\left(\ee^{-\frac{|\cdot|^2}{2z}}\ast f\right)(x)
\,dx\,dz 
\end{align}
(by Plancherel).
Here, the application of Fubini's theorem is justified because of
\begin{align}
&\int_{\R^3}\int_0^\infty \left|\overline{\hat g(k)}
k_ik_j \ee^{-\frac12 z|k|^2}\frac{z}{4}
\hat f(k)\right|\,dz\,dk
=
\int_{\R^3}\left|\overline{\hat g(k)}
\frac{k_ik_j}{|k|^4}
\hat f(k)\right|\,dk
\nonumber\\
\le& 
\int_{\R^3}\left|\overline{\hat g(k)}
\frac{1}{|k|^2}
\hat f(k)\right|\,dk
<\infty,
\end{align}
since $1/|k|^2$ is integrable near $0$ in $3$ dimensions.
We transform $t=\frac{1}{z}$, $z^{-\frac12}\,dz=-t^{-\frac32}\,dt$
and use 
\begin{align}
\partial_i\partial_j\ee^{-\frac{t|x|^2}{2}}=(t^2x_ix_j-t\delta_{ij})\ee^{-\frac{t|x|^2}{2}}
\end{align}
to obtain
\begin{align}
\operatorname{lhs} \eqref{eq:claim-partial2Delta-2}=&
\frac{1}{4(2\pi)^{\frac32}}
\int_0^\infty \int_{\R^3}
\overline{g(x)}
\partial_i\partial_j\left(\ee^{-\frac{t|\cdot|^2}{2}}\ast f\right)(x)
\,dx\,\frac{dt}{t^{\frac32}}
\nonumber\\
=&
\frac{1}{4(2\pi)^{\frac32}}
\int_0^\infty \int_{\R^3}\int_{\R^3}
\overline{g(x)}
(t^2(x_i-y_i)(x_j-y_j)-t\delta_{ij})\ee^{-\frac{t|x-y|^2}{2}} f(y)
\,dy\,dx\,\frac{dt}{t^{\frac32}}
\nonumber\\
=&
\frac{1}{4(2\pi)^{\frac32}}
\int_{\R^3}\int_{\R^3}
\overline{g(x)}
\int_0^\infty (t^2(x_i-y_i)(x_j-y_j)-t\delta_{ij})\ee^{-\frac{t|x-y|^2}{2}} 
\frac{dt}{t^{\frac32}}\,f(y)
\,dx\,dy.
\label{eq:calculation}
\end{align}
The application of Fubini's theorem in the last step is justified below.
Now, substituting $s=\frac{t}{2}|x|^2$, $t=\frac{2s}{|x|^2}$, 
$\sqrt{t}\,dt=2\sqrt{2}\sqrt{s}\frac{ds}{|x|^3}$,
$t^{-\frac12}\,dt=\frac{\sqrt{2}}{|x|}s^{-\frac12}\,ds$
and using $\Gamma(\frac12)=\sqrt{\pi}$, $\Gamma(\frac32)=\frac{\sqrt{\pi}}{2}$,
we calculate
\begin{align}
&\frac{1}{4(2\pi)^{\frac32}}\int_0^\infty (t^2x_ix_j-t\delta_{ij})\ee^{-\frac{t|x|^2}{2}} 
\frac{dt}{t^{\frac32}}
\nonumber\\=&
\frac{2\sqrt{2}}{4(2\pi)^{\frac32}}
\int_0^\infty \sqrt{s} \ee^{-s}\,ds\, \frac{x_ix_j}{|x|^3}
-\frac{\sqrt{2}}{4(2\pi)^{\frac32}}
\int_0^\infty \frac{1}{\sqrt{s}} \ee^{-s}\,ds\, \frac{\delta_{ij}}{|x|}
\nonumber\\=&
\frac{\Gamma(\frac32)}{\sqrt{2}(2\pi)^{\frac32}}
\frac{x_ix_j}{|x|^3}
-\frac{\Gamma(\frac12)}{2\sqrt{2}(2\pi)^{\frac32}}\frac{\delta_{ij}}{|x|}
\nonumber\\=&
\frac{1}{8\pi}\left(\frac{x_ix_j}{|x|^3}-\frac{\delta_{ij}}{|x|}\right)
\label{eq:calculation2}
\end{align}
and similarly
\begin{align}
\frac{1}{4(2\pi)^{\frac32}}\int_0^\infty \left|t^2x_ix_j-t\delta_{ij}\right|
\ee^{-\frac{t|x|^2}{2}} 
\frac{dt}{t^{\frac32}}
\le
\frac{1}{8\pi}\left(\frac{|x_ix_j|}{|x|^3}+\frac{\delta_{ij}}{|x|}\right)
\le\frac{1}{4\pi|x|},
\end{align}
which is integrable near $0$.
Together with the fast decay of $f$ and $g$, this bound justifies the 
application of Fubini's theorem in the last step of 
\eqref{eq:calculation}.
Substitution of the calculation \eqref{eq:calculation2} in 
\eqref{eq:calculation} proves claim~\eqref{eq:claim-partial2Delta-2}.
\end{proof}

Next, we derive a variant of the simplified dipole expansion presented in Section~\ref{sec:bounds-dipole-expansion} (using the notation from there), but now for the operator $\partial_i\partial_j\Delta^{-2}$ rather than $\Delta^{-1}$. 
This is used in the proof of Lemma~\ref{le:w-star-bd}. 
\begin{lemma}[Long-range asymptotics of $\partial_i\partial_j\Delta^{-2}$]
\label{le:dipole-expansion-wzwei}
For all $x\in\R^3$ with $|x|\ge 2R$ and $i,j,k\in[3]$, 
we have
\begin{align}
\label{eq:claim-dipole-simplified-applied}
\Delta^{-2}\partial_i\partial_j\rho(x)&=\frac{Q}{8\pi}\left(\frac{x_ix_j}{|x|^3}-\frac{\delta_{ij}}{|x|}\right)+r_1'(x),\\
\partial_k\Delta^{-2}\partial_i\partial_j\rho(x)&=\frac{Q}{8\pi} 
\left(\frac{\delta_{ik}x_j+\delta_{jk}x_i+\delta_{ij}x_k}{|x|^3}
-3\frac{x_ix_jx_k}{|x|^5}\right)
+\partial_kr_1'(x)
\end{align}
with the error bounds
\begin{align}
\label{eq:bound-partial-r1-strich}
|r_1'(x)|\le\frac{3M_1}{\pi}\frac{1}{|x|^2},\qquad
|\partial_k r_1'(x)|\le\frac{36M_1}{\pi}\frac{1}{|x|^3}.
\end{align}
\end{lemma}
\begin{proof}
Recall that for $x\in\R^3$ with $|x|\ge 2R$, $y\in\operatorname{supp}\rho$, and 
all $t\in[0,1]$, we have
$|x-ty|\ge |x|/2\ge R$.
We apply \eqref{eq:taylor1} to
\begin{align}
\label{eq:def-f-x-y-t}
f_{x,y}(t):=&\frac{(x_i-ty_i)(x_j-ty_j)}{|x-ty|^3}-\frac{\delta_{ij}}{|x-ty|},
\\
f_{x,y}(0)=&\frac{x_ix_j}{|x|^3}-\frac{\delta_{ij}}{|x|}, \\
\partial_k f_{x,y}(0)=& \frac{\delta_{ik}x_j+\delta_{jk}x_i+\delta_{ij}x_k}{|x|^3}
-3\frac{x_ix_jx_k}{|x|^5}.
\label{eq:partial-k-f}
\end{align}

We calculate:
\begin{align}
\label{eq:fxy'-new}
f_{x,y}'(t)=-\frac{y_i(x_j-ty_j)}{|x-ty|^3}
-\frac{(x_i-ty_i)y_j}{|x-ty|^3}
+3\frac{\scalara{y}{x-ty}(x_i-ty_i)(x_j-ty_j)}{|x-ty|^5}
-\delta_{ij}\frac{\scalara{x-ty}{y}}{|x-ty|^3},
\end{align}
\begin{align}
\partial_k f_{x,y}'(t)=&-\frac{y_i\delta_{jk}}{|x-ty|^3}
+3\frac{y_i(x_j-ty_j)(x_k-ty_k)}{|x-ty|^5}
-\frac{\delta_{ik}y_j}{|x-ty|^3}
+3\frac{(x_i-ty_i)y_j(x_k-ty_k)}{|x-ty|^5}
\nonumber\\&
+3\frac{y_k(x_i-ty_i)(x_j-ty_j)}{|x-ty|^5}
+3\frac{\scalara{y}{x-ty}\delta_{ik}(x_j-ty_j)}{|x-ty|^5}
+3\frac{\scalara{y}{x-ty}(x_i-ty_i)\delta_{jk}}{|x-ty|^5}\nonumber\\
& -15\frac{\scalara{y}{x-ty}(x_i-ty_i)(x_j-ty_j)(x_k-ty_k)}{|x-ty|^7}
-\delta_{ij}\left(\frac{y_k}{|x-ty|^3}
-3\frac{(x_k-ty_k)\scalara{x-ty}{y}}{|x-ty|^5}\right).
\end{align}
Using the Cauchy-Schwarz inequality, it follows
\begin{align}
|f_{x,y}'(t)|\le \frac{6|y|}{|x-ty|^2}\le\frac{24|y|}{|x|^2}, \qquad
|\partial_k f_{x,y}'(t)|\le \frac{36|y|}{|x-ty|^3}\le\frac{288|y|}{|x|^3}.
\end{align}
We obtain
\begin{align}
& \Delta^{-2}\partial_i\partial_j\rho(x) 
=\frac{1}{8\pi}\int_{\R^3} \frac{1}{|x-y|}
\left(\frac{(x_i-y_i)(x_j-y_j)}{|x-y|^2}-\delta_{ij}\right)\rho(y)\,dy
\nonumber\\&= 
\frac{1}{8\pi}\int_{\R^3} \left(f_{x,y}(0)+\int_0^1 f'_{x,y}(t)\,dt\right)
\rho(y)\,dy 
\label{eq:taylor-integrated-new}
\end{align}
and hence the claim~\eqref{eq:claim-dipole-simplified-applied} with an error term
$r_1'$ bounded by 
\begin{align}
|r_1'(x)|\le\frac{1}{8\pi}\int_{\R^3}\int_0^1\frac{24|y|}{|x|^2}
\,dt\,|\rho(y)|\,dy
=\frac{3M_1}{\pi}\frac{1}{|x|^2}.
\end{align}
This shows the first claimed bound in~\eqref{eq:bound-partial-r1-strich}.

For the partial derivatives $\partial_k$ 
with respect to $x_k$ we obtain:
\begin{align}
& \partial_k\Delta^{-2}\partial_i\partial_j\rho(x)  
= \frac{1}{8\pi}\int_{\R^3} \left(\partial_k f_{x,y}(0)
+\int_0^1 \partial_k f'_{x,y}(t)\,dt\right)
\rho(y)\,dy .
\end{align}
Consequently,
\begin{align}
|\partial_k r_1'(x)|\le\frac{1}{8\pi}\int_{\R^3}\int_0^1\frac{288|y|}{|x|^3}
\,dt\,|\rho(y)|\,dy
=\frac{36M_1}{\pi}\frac{1}{|x|^3}.
\end{align}
This proves the second bound in~\eqref{eq:bound-partial-r1-strich}.
\end{proof}

\subsubsection{Integral kernels close to the diagonal}
\label{sec:appendix-kernels-diag}

Recall that $\|\tilde J\|_1$ denotes the total unsigned mass of 
a signed measure $\tilde J$.

\begin{lemma}[Uniform bounds]
\label{le:uniform-bounds}
For the form function $\varphi$ introduced in 
Subsection~\ref{subsubsec:model-assum}, one has 
\begin{align}
\label{eq:bound-partial-delta-phi}
\csiebzehn := & \max_{l\in[3]}\|\partial_l\Delta^{-1}\varphi\|_\infty<\infty, \\
\cachtzehn := & \max_{l,i,j\in[3]}\|\partial_l\partial_i\partial_j\Delta^{-2}\varphi\|_\infty<\infty .
\label{eq:bound-partial-dreifach-delta-phi}
\end{align}
Additionally, for any signed measure $J$ with $\|J\|_1<\infty$ 
and for all $l,i,j\in[3]$, one has 
\begin{align}
\label{eq:bound-partial-delta-phiJ}
& \|\partial_l\Delta^{-1}\varphi*J\|_\infty\le\csiebzehn\|J\|_1,  \\
& \|\partial_l\partial_i\partial_j\Delta^{-2}\varphi*J\|_\infty\le\cachtzehn\|J\|_1. 
\label{eq:bound-partial-dreifach-delta-phiJ}
\end{align}
\end{lemma}
\begin{proof}
The operators
$\partial_l\Delta^{-1}$ and $\partial_l\partial_i\partial_j\Delta^{-2}$ have the 
integral kernels $k_1(x-y)$ and $k_2(x-y)$, respectively, where 
\begin{align}
k_1(z):=\frac{1}{4\pi}\frac{z_l}{|z|^3}\quad\text{ and }\quad
k_2(z):=\frac{1}{8\pi}\left(
\frac{\delta_{il}z_j+\delta_{jl}z_i+\delta_{ij}z_l}{|z|^3}
-3\frac{z_iz_jz_l}{|z|^5}\right), 
\end{align}
cf.\ \eqref{eq:claim-partial2Delta-2}, \eqref{eq:def-f-x-y-t}, and \eqref{eq:partial-k-f}. Since both kernels are bounded by 
$O(|z|^{-2})$, they are locally integrable
and decay at infinity. Because $\varphi$ is compactly supported and bounded, 
the claims \eqref{eq:bound-partial-delta-phi} and 
\eqref{eq:bound-partial-dreifach-delta-phi} follow. The remaining claims
follow from 
$\|k*J\|_\infty\le\|k\|_\infty\|J\|_1$ with $k=k_1$ and $k=k_2$, respectively.
\end{proof}
\end{appendix}

\subsection*{Acknowledgment}
The authors would like to thank an anonymous referee for constructive comments 
helping us to improve the paper. Visits of RB to Munich were supportedby LMU’s 
Center for Advanced Studies.

\end{document}